\documentclass[journal]{IEEEtran}
\usepackage[numbers,sort&compress]{natbib}
\usepackage{graphicx}
\usepackage{amsmath}
\usepackage{amssymb}
\usepackage{amsthm}
\usepackage{epsfig}
\usepackage{epstopdf}

\usepackage{algorithm,algorithmic}

\usepackage{mathptmx}
\usepackage{color}
\usepackage{multirow}
\usepackage{bm}
\usepackage{ulem}
\usepackage{mathtools}

\usepackage{cuted}

\usepackage{xcolor}
\usepackage{caption}
\usepackage{eucal}

\usepackage{tabularx,booktabs,ragged2e}
\newtheorem{assumption}{Assumption}

\allowdisplaybreaks[3]

\ifCLASSOPTIONcompsoc
\usepackage[caption=false,font=normalsize,labelfont=sf,textfont=sf]{subfig}
\else
\usepackage[caption=false,font=footnotesize]{subfig}
\fi

\ifCLASSINFOpdf

\else

\fi

\newcommand{\matA}{\boldsymbol{A}}
\newcommand{\matB}{\boldsymbol{B}}

\newcommand{\matP}{\boldsymbol{p}}
\newcommand{\matC}{\boldsymbol{c}}

\newcommand{\matPhi}{\boldsymbol{\phi}}

\newcommand{\bs}[1]{\pmb{#1}}

\newcommand{\diag}[1]{\text{diag}\left(#1\right)}
\newcommand{\rank}[1]{\text{rank}\left(#1\right)}
\newcommand{\Tr}[1]{\text{Tr}\left(#1\right)}
\newcommand{\LRT}[2]{%
	\mathrel{\mathop\gtrless\limits^{#1}_{#2}}%
}

\newcommand{\fracPS}[1]{\frac{1}{\sigma_{#1}^2}}

\newcommand{\vecwang}[1]{\text{vec}\left( #1 \right)}
\newtheorem{theorem}{Theorem}

\newtheorem{definition}{Definition}
\newtheorem{lemma}{Lemma}

\definecolor{zswPink}{RGB}{169, 53, 41}
\definecolor{ao(english)}{rgb}{0.0, 0.5, 0.0}
\definecolor{armygreen}{rgb}{0.29, 0.33, 0.13}
\definecolor{chocolate(web)cocoabrown}{rgb}{1.0, 0.75, 0.0}
\definecolor{buff}{rgb}{0.94, 0.86, 0.51}
\definecolor{chocolate(web)}{rgb}{0.82, 0.41, 0.12}
\definecolor{cocoabrown}{rgb}{0.82, 0.41, 0.12}

\newcommand{\inchhspace}{\hspace{1in}}

\DeclarePairedDelimiter\abs{\lvert}{\rvert}%
\makeatletter
\let\oldabs\abs
\def\abs{\@ifstar{\oldabs}{\oldabs*}}

	\definecolor{ao(english)}{rgb}{0.0, 0.5, 0.0}

\newcommand{\hpt}{\hspace{2pt}}
\newcommand{\rewCone}{}

\begin{document}

\title{A Statistical Learning-Based Algorithm for Topology Verification in Natural Gas Networks  Based on Noisy Sensor Measurements  \thanks{This work was supported by the Department of Energy under Award DEOE0000779. } }

\author{Zisheng Wang and Rick S. Blum, \IEEEmembership{Fellow, IEEE}.%
	\thanks{Zisheng Wang and Rick S. Blum are with the Department of Electrical and Computer Engineering, Lehigh University, Bethlehem, PA 18015 USA (e-mail: \{ziw417,rblum\}@lehigh.edu)}}
\maketitle

\IEEEpeerreviewmaketitle

%


\begin{abstract}
Accurate knowledge of natural gas network topology is critical for the proper operation of natural gas networks. Failures, physical attacks, and cyber attacks can cause the actual natural gas network topology to differ from what the operator believes to be present. Incorrect topology information misleads the operator to apply inappropriate control causing damage and lack of gas supply. Several methods for verifying the topology have been suggested in the literature for electrical power distribution networks, but we are not aware of any publications for natural gas networks. In this paper, we develop a useful topology verification algorithm for natural gas networks based on modifying a general known statistics-based approach to eliminate serious limitations for this application while maintaining good performance. We prove that the new algorithm is equivalent to the original statistics-based approach for a sufficiently large number of sensor observations.  We provide new closed-form expressions for the asymptotic performance that are shown to be accurate for the typical number of sensor observations required to achieve reliable performance. 
\end{abstract}

\begin{IEEEkeywords}
	Natural gas networks, generalized likelihood ratio test, semidefinite relaxation programming, asymptotic performance.
\end{IEEEkeywords}

\section{Introduction}

Modern natural gas delivery monitoring and control systems incorporate a Supervisory Control and Data Acquisition (SCADA) system for enhanced remote wide area control and situational awareness. However, the increased usage of information and communication technologies in these SCADA systems also introduces more vulnerability into natural gas networks by providing opportunities for malicious attackers. Attacks can, for example, modify the system databases that hold the current natural gas network topology, which is essential for the operator to evaluate and control the network. The operator might also obtain incorrect knowledge of the topology from modified communications or actual physical attacks which modify the topology. 
Incorrect topology information can cause an operator to apply inappropriate control that may cause severe damage to the system and lack of supply to critical natural gas fired electric power generation plants. This can cause a loss of power to critical infrastructure. In this paper, we study topology verification which attempts to decide if the topology the operator believes to be present is actually present. We use a hypothesis testing formulation based on noisy sensor measurements.

A number of publications have focused on topology verification for electrical networks \cite{weimer2012distributed,liang2017framework,sevlian2018outage,cavraro2015distribution}. Other publications on the related topic of topology identification have also appeared \cite{deka2016estimating,gao2013method,tian2016mixed,deka2018structure}. In topology identification, there is no assumption that the operator believes some particular topology is present. In \cite{weimer2012distributed}, the authors provide centralized and distributed algorithms  for topology verification for electrical networks using noisy sensor measurements.
In \cite{liang2017framework}, the authors document  the impact of the operator believing an incorrect topology is present for electrical network cases.
The authors of \cite{sevlian2018outage} propose a maximum likelihood outage detection framework for tree  structured electrical power distribution networks using real time power flow measurements and load forecasts. 
The authors of \cite{cavraro2015distribution} propose a topology identification algorithm for electrical networks by employing time series sensor measurements. In \cite{deka2016estimating}, the authors present a framework for topology identification for radial structured electrical networks. The authors of \cite{gao2013method} develop an algorithm for identifying the correct topology  of an electrical network by using power measurements in a novel manner. In \cite{tian2016mixed}, the authors propose a mixed integer quadratic programming model to identify the actual topology  of an electrical network with noisy sensor measurements. The authors of \cite{deka2018structure} develop a low-complexity algorithm to learn the current topology of an electrical network using measurements from a subset of all nodes.

While some topology verification and identification algorithms have been reported in the literature for electrical networks \cite{weimer2012distributed,liang2017framework,sevlian2018outage,cavraro2015distribution,deka2016estimating,gao2013method,tian2016mixed,deka2018structure}, we did not see any publications providing such algorithms for natural gas networks. 
Further, the physics-based mathematical equations that the natural gas network sensor measurements obey are very different from the equations that electrical network sensor measurements follow. Since the published algorithms for electrical networks are highly dependent on the electrical network physics-based mathematical equations, these algorithms can't be applied in natural gas networks. Therefore, in this paper we design a new topology verification algorithm for natural gas networks. This topology verification algorithm is obtained by modifying a general known statistics-based approach \cite{lehmann2006testing} consisting of non-convex optimization problems. We rigorously justify that the new algorithm has the same performance as the original statistic-based approach when the number of sensor observations satisfies a given condition. Moreover, the new algorithm is more efficient in terms of run time while also being more reliable than the original statistic-based approach. We provide new closed-form expressions for the asymptotic performance that are shown to be accurate for the typical number of sensor observations required to achieve reliable performance.
These asymptotic expressions are analytically justified.  We have not seen any derivations in the literature of the asymptotic performance for the cases we consider in this paper, which involve estimations of discrete variables which describe the natural gas network topology. 

\subsection{Notation and organization}

Throughout this paper, bold upper case letters denote matrices and bold lower case letters denote vectors. Let $ [\bs{A}]_{m,n} $ denote the element in the $ m^{th} $ row and $ n^{th} $ column of the matrix $ \bs{A} $. Let $ [\bs{A}]_{[i,:]} $ denote the $ i^{th} $ row vector of the matrix $ \bs{A} $. The quantity $ \vecwang{\bs{A}} $ denotes the vector of the stacked columns of the matrix $ \bs{A} $. $ [\bs{A}]_+ = \max(0,\bs{A}) $ where the max is applied element-wise. Let $ \Tr{\bs{A}} $ denote the trace of the matrix $ \bs{A} $. A list of critical notations is shown in Table \ref{tb:notation_list}.

\begin{table}
	\renewcommand{\arraystretch}{1.3}
	\newcolumntype{L}{>{\raggedright\arraybackslash}X} 
	\caption{The list of critical notations}
	\label{tb:notation_list}
	\centering
	\begin{tabularx}{0.495\textwidth}{lL}
		\hline\hline
		Notation & Description \\
		\hline	
		$ \bs{\phi},\bs{\phi}_F,\bs{\phi}_C $& Gas flow vectors for all pipelines, fixed pipelines, and changeable pipelines.\\
		$ \bs{p}  $&   The gas pressure vector.\\
		$ \bs{q}$&    The gas injection/withdraw vector.\\
		$ \tilde{\bs{p}}(t),\tilde{\bs{q}}(t),\tilde{\bs{\phi}}(t) $& The sensor measurements of the pressures, the injections, and the flows.\\
		$ \bs{v} $& The full vector of observations.\\
		$ T_a $& The number of observations.\\
		$ L,N $ &	The number of pipelines and nodes (excluding the reference node).\\
		$ \tilde{\bs{A}},\bs{a},\bs{A} $& The incidence matrix for the topology, the column vector of the incidence matrix corresponding to the reference node, the rest of the incidence matrix with $ \bs{a} $ removed.\\
		$ \bs{A}_{H_0},\bs{A}_{H_1},\bs{A}_{FC} $& The incidence matrices for the topology which the operator believe to be true, the actual topology, and the topology with all changeable pipelines closed.\\ 
		$ \bs{B},\bs{b} $& Weighted incidence matrices.\\
		$ \bs{\delta}_p,\bs{\delta}_q,\bs{\delta}_\phi,\bs{\delta} $& The sensor placements variables for the pressure, the injection, and the flow sensors.\\ 
		$ \bs{\theta},\hat{\bs{\theta}},\bs{\theta}^* $& The parameter of the hypothesis testing problem, the estimated parameter, the computed parameter by using data from $ H_1 $.\\
		$ \bs{\omega},\hat{\bs{\omega}},\bs{\omega}^* $& Continuous part of $ \bs{\theta} $, $ \hat{\bs{\theta}} $, and $ \bs{\theta}^* $.\\
		$ g(\bs{v}|\bs{\theta}) $& The pdf of $ \bs{v} $ parameterized by $ \bs{\theta} $.\\
		$ \bs{X} $& The variables for the Relaxed GLRT in $ \mathbb{R}^{S\times S} $.\\
		$ \bs{M},\bs{Z} $& The matrices used in the objective function and constraints of the Relaxed GLRT.\\
		$ \bs{J}(\bs{A}) $& The fisher information matrix for estimating $ \bs{\omega} $ given a general $ \bs{A} $.\\
		$ \bs{f}(\bs{\theta}) $& The constraint functions.\\
		$ \bs{F}(\bs{A},\bs{\omega}), \widetilde{\bs{F}} $& The gradient of $ \bs{f}(\bs{\theta}) $ with respect to $ \bs{\omega} $ and the gradient when $ \bs{A} = \bs{A}_{FC} $.\\
		$ \bs{U}(\bs{A},\bs{\omega}),\widetilde{\bs{U}} $& Orthogonal bases for null spaces of $ \bs{F}(\bs{A},\bs{\omega}) $ and $ \widetilde{\bs{F}} $.\\
		$ \lambda $&  The non-centrality parameter of chi-squared distribution.\\
		\hline
		\hline
	\end{tabularx}
\end{table} 

The remainder of this paper is organized as follows. Section II describes the topology verification problem based on sensor observations in a natural gas network. Section III describes our statistics-based algorithm to decide if the topology the operator believes in agrees with the true one. 
A minimum requirement on a proper sensor placement is also described. Section IV describes our closed-form asymptotic performance predictor. Section V presents numerical results and analysis. Section VI presents conclusions and discussions. 

This paper is an extended version of the preliminary work presented at an IEEE conference \cite{wang2019topology}. In \cite{wang2019topology} we proposed a relaxed algorithm to verify the topology of a natural gas network without any analytical justification of good performance. Here we prove the relaxation loses no performance for a sufficient number of observations in Section III. We develop closed-form performance expressions which are shown to be accurate for the desired performance levels in Section IV. We also provide conditions on a suitable sensor placement in Section IV and include greatly expanded numerical results in Section V. 

\section{Natural Gas Network Model}\label{section_2}

The natural gas network is modeled as a directed graph $ \mathcal{G}\triangleq(\Xi,\mathcal{L}) $. The set of graph vertices $ \Xi = \{0,1,\cdots,N\} $ represents all possible nodes in the natural gas network.
At each node, gas is possibly injected or withdrawn. { Gas injection occurs when gas suppliers produce gas or gas is provided by adjacent networks. On the other hand, gas withdraw occurs when gas is consumed or sent to adjacent networks. } 
Let $ q_i $ and $ p_i $ represent the gas injection and pressure at node $ i $ respectively.  We use the term pipeline to refer to a pipe connecting two adjacent nodes. 
The gas pipeline between two nodes is modeled by one of the graph edges in the set $ \mathcal{L} = \{1,2,\cdots,L\} $. The gas flow through the $ l^{th} $ pipeline is $ \phi_l $. The gas injection $ q_i $ at node $ i $ is positive for injection, negative for withdraw, and zero if no injection or withdraw takes place. If the direction of the gas flow $ \phi_l $ is the same as the direction of the directed graph, $ \phi_l > 0$, otherwise, $ \phi_l < 0 $. 

Fig. \ref{fig_example_2_model} illustrates a gas distribution system modeled using a directed graph. We allow the topology verification algorithm we describe in this paper to use pressure, injection/withdraw, and flow measurements. In Fig. \ref{fig_example_2_model}, these three possible types of sensors are illustrated.
Fig. \ref{fig_example_2_model} shows the sensors at node 1 and the pipeline directly connected to node 1. At the other nodes and pipelines, sensors are employed in the similar way. 
It is not necessary to employ all three types of sensors at every node or pipeline. We will describe the strategy to employ sensors later in this paper.
\begin{figure}
	\centering
	\includegraphics[width=0.49\textwidth]{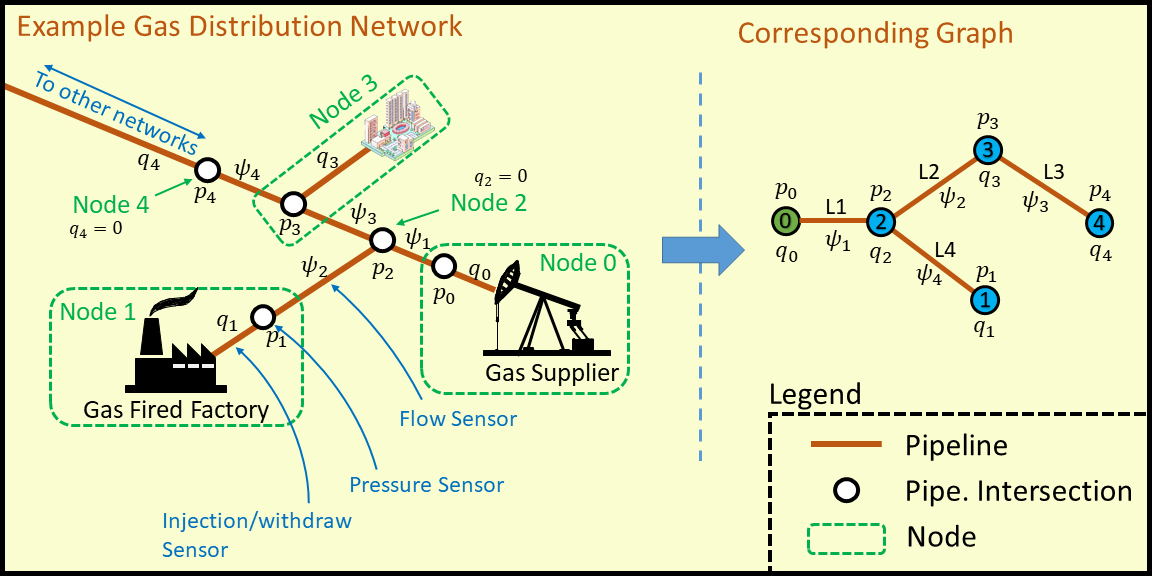}
	\caption{An example gas distribution network and the corresponding directed graph.}
	\label{fig_example_2_model}
\end{figure}

Consider a medium to high pressure network with time-invariant gas injections.
The relationship between the gas node pressures $ p_i,  p_j $ and the gas flow $ \phi_l $ along the pipeline $l$ connecting nodes $i$ and $j$ follows the Weymouth equation \cite{osiadacz1987simulation}
\begin{align}
\alpha_lp_i^2-p_j^2=c_l\phi_l|\phi_l|,\inchhspace l = 1,\ldots, L  \label{weymouth}
\end{align}
when the graph edge is directed from $ i $ to $ j $ and $ \alpha_l $ is the compressor ratio in the $ l $th pipeline. Let $ c_l = \gamma_l[1-(1-\alpha_l)r_l]>0 $ denote the characteristic value of $ l^{th} $ pipelines where $ \gamma_l $ is the pipeline constant and $ r_l\in[0,1] $ indicates the compressor  is located at normalized distance $ r_l $ from node $ i $ and $ (1-r_l) $ from node $ j $. When no compressor is present at the $ l^{th} $ pipeline, $ \alpha_l = 1 $. The natural gas network satisfies the mass-conservation equation. At each node $ i $, the gas injection $ q_i $ is equal to the sum of gas flows that leave node $ i $ minus the sum of gas flows that enter node $ i $, such that
\begin{align}
&q_i = \sum_{l:\text{leave }i}\phi_l - \sum_{l:\text{enter }i}\phi_l,\inchhspace i = 0,\cdots,N \label{gas_mass_conversation_1} 
\end{align}
and
\begin{align}	
&q_0 = -\sum_{i\in\Xi,i\ne 0}q_i	\label{gas_mass_conversation_2}
\end{align}
where $q_0$ is the injection from the reference node. The reference node is typically chosen as a node supplying gas.

Define $ \bs{p}\triangleq[ p_1,\hpt\cdots,\hpt p_N ]^T $ and $ \bs{q}\triangleq[ q_1,\hpt\cdots,\hpt q_N ]^T $. Note that the reference node variables $ p_0 $ and $ q_0 $ (assumed known) are excluded from the vectors $ \bs{p} $ and $ \bs{q} $. Similarly, we also construct $ \bs{\phi}\triangleq[ \phi_1,\hpt\cdots,\hpt\phi_L ]^T $, $ \bs{c}\triangleq[ c_1,\hpt\cdots,\hpt c_L ]^T $ and $ \bs{\alpha}\triangleq[ \alpha_1,\hpt\cdots,\hpt\alpha_L ]^T $. 
Based on (\ref{weymouth}), the $ L\times (N+1) $ incidence matrix $ \bs{\tilde{A}} $ has elements
\begin{align}
&[\tilde{\bs{A}}]_{l,n} = \begin{cases}
+1, & n = i+1 \\
-1, & n = j+1 \\
0, & \text{otherwise}
\end{cases},\label{incidence_matrix}
\end{align}
with $ l = 1,2,\dots,L, \ n = 0,1,\dots,N $ and $ 0\le i,j \le N $. 
Isolating the first column corresponding to the reference node, the matrix $ \bs{\tilde{A}} $ can be partitioned as $ \bs{\tilde{A}} = \begin{bmatrix} \bs{a} & \bs{A} \end{bmatrix} $. From  (\ref{gas_mass_conversation_1}) we obtain $ \bs{q} = \bs{A}^T\bs{\phi} $ which does not involve the 
reference node.   A similar set of manipulations provides a matrix and vector version of (\ref{weymouth}) that is independent of the reference node, see \cite{osiadacz1987simulation}, to obtain 
\begin{align}
\matA^T\matPhi&=\bs{q}, \label{PF_1}\\
\matB(\matP\odot\matP)&=\matC\odot\matPhi\odot|\bs{\phi}|-p_0^2\bs{b}, \label{PF_2}
\end{align}
where $ \odot $ denotes the element-wise product, 
\begin{align}
\matB&\triangleq \text{diag}\{\boldsymbol{\alpha}\}[\matA]_+-[-\matA]_+, \label{def_B}\\
\bs{b}&\triangleq \text{diag}\{\boldsymbol{\alpha}\}[\boldsymbol{a}]_+-[-\boldsymbol{a}]_+, \label{Defb}
\end{align}
and $ [x]_+ = \max(0,x) $ is an element wise operator.

\section{Topology  verification } \label{sec_topo_attack_detect}

Let $ \Psi\subset\mathcal{L} $ denote the flow measurements employed by our algorithm. 
Let $ \Gamma\subset\Xi $ denote the pressure measurements employed, while $ \Lambda\subset\Xi $ denotes the  injection/withdraw measurements employed.  Assume a specific set of pipelines are changeable such that their states can be open (not employed) or closed (employed).  To be completely general, we allow all pipelines to be changeable but in practice this is usually not the case.  Let $ L_C $ denote the number of changeable pipelines, while $ \bs{\phi}_C $ is the vector of flows over these changeable pipelines. The rest of the pipelines are called fixed pipelines. Let $ L_F $ denote the number of fixed pipelines, while $ \bs{\phi}_F $ is the vector of flows over these pipelines. We must have $ L_F+L_C = L $ total pipelines. We also define four indicator functions that describe if measurements from the $ {\rewCone u}^{th} $ fixed pipeline, $ k^{th} $ changeable pipeline, $ n^{th} $ pressure sensor, or $ n^{th} $ injection  sensor will be employed as
\begin{align}
\rewCone {\delta}_{F,u} &= \begin{cases}
1, & \text{if the $ u^{th} $ fixed pipeline}\in\Psi\\
0, & \text{otherwise}
\end{cases},\label{censoring_vector_1}\\
{\delta}_{C,k} &= \begin{cases}
1, & \text{if the $ k^{th} $ changeable pipeline}\in\Psi\\
0, & \text{otherwise}
\end{cases},\label{censoring_vector_2}\\
{\delta}_{p,n} &= \begin{cases}
1, & \text{if the $ n^{th} $ pressure sensor}\in\Gamma\\
0, & \text{otherwise}
\end{cases},\label{censoring_vector_3}
\end{align} 
and
\begin{align}
{\delta}_{q,n} &= \begin{cases}
1, & \text{if the $ n^{th} $ injection sensor}\in\Lambda\\
0, & \text{otherwise}
\end{cases},\label{censoring_vector_4}
\end{align}
with $ {\rewCone u = 1,2,\dots,L_F}, k=1,2,\dots, L_C, $ and $ n = 1,2,\dots,N. $
Let $ \bs{\delta}_F = (\delta_{F,1},\dots,\delta_{F,L_F})^T,\hpt\bs{\delta}_C = (\delta_{C,1},\dots,\delta_{C,L_C})^T,\hpt\bs{\delta}_p = (\delta_{p,1},\dots,\delta_{p,N})^T,\bs{\delta}_q = (\delta_{q,1},\dots,\delta_{q,N})^T $, $ \bs{\delta}_\phi = (\bs{\delta}_F^T,\bs{\delta}_C^T) $ and $ {\bs{\delta}} =(\bs{\delta}_F^T,\allowbreak\bs{\delta}_C^T,\allowbreak\bs{\delta}_p^T,\allowbreak\bs{\delta}_q^T)^T $ denote the full vector of indicator functions.
The following assumptions are made through out this paper.
\begin{assumption} \label{Assumption_1}
	The set of all nodes  which could be employed is assumed to be known to the operator correctly. 
\end{assumption}
\begin{assumption} \label{Assumption_2}
	The set of all measurement errors in the gas pressures, flows and injections/withdraws forms an independent distributed sequence of zero-mean Gaussian random variables and are not manipulated by attackers. The measurement errors are taken to be relatively small in this study as they would be in practice. Larger measurement errors still follow the trends presented. 
\end{assumption}

{\bf Assumptions} \ref{Assumption_1} and \ref{Assumption_2} define our assumed problem model.   In {\bf Section~\ref{relaxation_method}}, we provide minimum requirements on the sensor placement employed. 
Recall that the matrix $ \bs{A} $ describes the topology of the network when the reference node is excluded. Each row of $ \bs{A} $ describes which two nodes are connected by the pipeline corresponding to the row number as per (\ref{incidence_matrix}). Let $ \bs{A}_{H_0} $ denote the topology describing matrix $ \bs{A} $ that the operator believes is correct. For convenience, we assume the first $ L_F $ rows of $ \bs{A} $, denoted by  $ \bs{A}_F $, correspond to the fixed pipelines while the next $ L_C $ rows, denoted by  $ \bs{A}_C $, correspond to the changeable pipelines. Thus $\bs{A} = \begin{bmatrix}\bs{A}_F^T&\bs{A}_C^T\end{bmatrix}^T $. 
The same partition is applied to $ \bs{A}_{H_0} $, $\bs{a} $, $ \bs{B} $, $ \bs{b} $, and $ \bs{\phi} $

In order to verify the topology, we use sensor measurements. 
At times $ t = 1,2,\cdots, T_a $, we observe noisy versions of the deterministic gas pressures $ {{\bs{p}}}[t]$, the injections/withdraws $ {\bs{q}}[t] $, the gas flows over the fixed pipelines $ {\bs{\phi}}_F[t] $ and the gas flows over the changeable pipelines ${\bs{\phi}}_C[t] $. 
Let $ \tilde{\bs{p}}[t],\tilde{\bs{q}}[t], \tilde{\bs{\phi}}_F[t] $ and $ \tilde{\bs{\phi}}_C[t] $ denote these noisy observations for $ t = 1,2,\cdots, T_a $. 
By {\bf Assumption} \ref{Assumption_2}, at each time $ t $ these observations follow a jointly Gaussian distribution with independent components
\begin{align}
\tilde{\bs{p}}[t]&\sim\mathcal{N}\left( \bs{p},\sigma_{p}^2\bs{I}_N \right),\label{gaussian_p} \\
\tilde{\bs{q}}[t]&\sim\mathcal{N}\left(\bs{q},\sigma_{q}^2\bs{I}_N \right), \label{gaussian_q}
\end{align}
and
\begin{align}
\tilde{\bs{\phi}}[t] = \begin{bmatrix}
\tilde{\bs{\phi}}_F[t]\\\tilde{\bs{\phi}}_C[t]
\end{bmatrix}\sim\mathcal{N}\left(\begin{bmatrix}
\bs{\phi}_F\\\bs{\phi}_C
\end{bmatrix},\sigma_{\phi}^2\bs{I}_L \right)\label{gaussian_phi_1},
\end{align}
where $ \bs{I}_N $ denotes an $ N\times N$ identity matrix, $ \sigma_p^2,\sigma_q^2 $ and $ \sigma_{\phi}^2 $ denote the 
variances of the measurement errors of the observed gas pressures, withdraws and flows, respectively. 

Let $ P_G(\bs{x}|\bs{\mu},\bs{\Sigma}) $ denote a Gaussian pdf with the mean vector $ \bs{\mu} $ and the covariance matrix $ \bs{\Sigma} $. Denote the full vector of observations as $ \bs{v} = ( \tilde{\bs{p}}[t],\tilde{\bs{\phi}}_F[t], \tilde{\bs{\phi}}_C[t],\tilde{\bs{q}}[t] ) $ and denote its pdf as $g(\bs{v} | \bs{\theta})$ which is parameterized by $ \bs{\theta} = ( \vecwang{\bs{A}}^T,\bs{p}^T,\bs{\phi}_F^T,\bs{\phi}_C^T )^T $.
Based on {\bf Assumption} \ref{Assumption_2} and equations (\ref{gaussian_p})-(\ref{gaussian_phi_1}), we have 
\begin{align}
&g(\bs{v} | \bs{\theta}) = \prod_{t=1}^{T_a}\left[ \prod_{\begin{subarray}{c} n=1\\\bs{\delta}_{p,n}\ne 0\end{subarray}}^Np_G(\tilde{\bs{p}}_n[t]|\bs{p}_n,\sigma_p^2)\prod_{\begin{subarray}{c} n=1\\\bs{\delta}_{q,n}\ne 0\end{subarray}}^Np_G(\tilde{\bs{q}}_n[t]|[\bs{A}^{T}\bs{\phi}]_n,\sigma_{q}^2)\right.\notag\\
&\left.\prod_{\begin{subarray}{c} u=1\\\bs{\delta}_{F,u}\ne 0\end{subarray}}^{L_F}p_G(\tilde{\bs{\phi}}_{F,u}[t]|\bs{\phi}_{F,u},\sigma_{\phi}^2)
\prod_{\begin{subarray}{c} k=1\\\bs{\delta}_{C,k}\ne 0\end{subarray}}^{L_C}p_G(\tilde{\bs{\phi}}_{C,k}[t]|\bs{\phi}_{C,k},\sigma_{\phi}^2) \right]. \label{e14} 
\end{align}
To test if the topology matches what the operator believes to be true (verification), we will test if the observed data comes from one of two different possible sets of pdfs. The first set includes all pdfs $ g(\bs{v}|\bs{\theta}) $ with $ \bs{\theta} $ such that $ {\bs{A}}={\bs{A}_{H_0}} $ where $ \bs{A}_{H_0} $ corresponds to the topology the operator expects to be present. The other set includes all other possible pdfs. 
Due to the fact that the gas pressures and the flows have to follow (\ref{PF_2}), then $ \bs{\theta} $ must satisfy
\begin{align}
\bs{f}(\bs{\theta}) = \bs{B}(\bs{p}\odot\bs{p}) - \bs{c}\odot\begin{bmatrix}
\bs{\phi}_F\\
\bs{\phi}_C
\end{bmatrix}\odot\begin{bmatrix}
\bs{\phi}_F\\
\bs{\phi}_C
\end{bmatrix} + \bs{b}p_0^2 = \bs{0}, \label{con_func_stan_GLRT}
\end{align}
in which $ \bs{B} $ and $ \bs{b} $ were defined in (\ref{def_B}) and (\ref{Defb}). 
Define the parameter sets $ \Theta_{H_0} $ and $ \Theta_{H_1} $ as
\begin{align}
\Theta_{H_0} = \left\{ \bs{\theta};\hspace{5pt}\bs{f}(\bs{\theta}) = \bs{0},\hspace{5pt}\bs{A} = \bs{A}_{H_0} \right\}, \label{def_par_set_h_0}
\end{align}
and
\begin{align}
\Theta_{H_1} = \left\{ \bs{\theta};\hspace{5pt}\bs{f}(\bs{\theta}) = \bs{0},\hspace{5pt}\bs{A} \ne \bs{A}_{H_0} \right\} \label{def_par_set_h_1}.
\end{align}

Verifying a topology requires solving the following hypothesis test involving the family of pdfs $ g(\bs{v} | \bs{\theta})$ parameterized by $ \bs{\theta} = (\vecwang{\bs{A}}^T,\hpt\bs{p}^T,\hpt\bs{\phi}_F^T,\hpt\bs{\phi}_C^T
)^T $ 
\begin{align}
H_0: \bs{\theta} \in \Theta_{H_0}\hpt\text{ and }\hpt  H_1: \bs{\theta} \in \Theta_{H_1}, \label{def_Set_H1}
\end{align}
where in either case the parameter must satisfy the constraint in (\ref{con_func_stan_GLRT}).
Both $ H_0 $ and $ H_1 $ are composite hypotheses (more than one pdf is possible under $ H_0 $ and $H_1$) which generally makes it much more difficult to define an optimum test \cite{poor2013introduction}. 
For a specific set of problems \cite{zeitouni1992generalized}, the Generalized Likelihood Ratio Test (GLRT) provides the largest error exponent under $ H_1 $ such that the probability of error under $ H_1 $ 
decreases most rapidly with more observations. 
The popular GLRT makes its decision by comparing a test statistic to a threshold $ \rho $ as
\begin{align}
&\frac{\sup_{\bs{\theta}\in\Theta_{H_1}} g(\bs{v}|\bs{\theta})}{\sup_{\bs{\theta}\in\Theta_{H_0}} g(\bs{v}|\bs{\theta})}\LRT{H_1}{H_0}\rho.\label{GLRT_test}
\end{align}
The threshold $ \rho $ is typically set to fix the probability that we decide for $ H_1 $ given that $ H_0 $ is true. This probability is called the false alarm probability and denoted by $ P_{fa} $. The optimizations in \eqref{GLRT_test} are called constrained maximum likelihood estimates of $ \bs{\theta} $. The probability of deciding for $ H_1 $ given $ H_1 $ is actually true is called the detection probability and denoted by $ P_d $.
The test in (\ref{GLRT_test}) is equivalent to
\begin{align}
&\sup_{\bs{\theta}\in\Theta_{H_1}} \ln g(\bs{v}|\bs{\theta}) - \sup_{\bs{\theta}\in\Theta_{H_0}} \ln g(\bs{v}|\bs{\theta})\LRT{H_1}{H_0}\ln \rho,\label{log_GLRT_test}
\end{align}
in which
\begin{align}
&\ln g(\bs{v}|\bs{\theta})
= C - \sum_{t = 1}^{{T_a}}\frac{1}{2} \left[ \frac{1}{\sigma_\phi^2}\left\|\begin{pmatrix}
\bs{\delta}_F\\\bs{\delta}_C
\end{pmatrix}\odot\left[\begin{pmatrix}
\tilde{\bs{\phi}}_F[t]\\\tilde{\bs{\phi}}_C[t]
\end{pmatrix} - \begin{pmatrix}
\bs{\phi}_F\\\bs{\phi}_C
\end{pmatrix}\right]\right\|^2 + \right.\notag\\
&\left.\frac{1}{\sigma_p^2}\left\|\bs{\delta}_p\odot\left(\tilde{\bs{p}}[t] - \bs{p}\right)\right\|^2 + \frac{1}{\sigma_q^2} \left\|\bs{\delta}_q\odot\left(\tilde{\bs{q}}[t] - \bs{A}^T\bs{\phi}\right)\right\|^2\right], \label{likelihood_f}
\end{align}
where $ \|\bullet\| $ denotes the $ l^2 $ norm, and 
$C = -{T_a}\|\bs{\delta}\|^2\ln(2\pi)/2-{T_a}\{\|\bs{\delta}_q\|^2\ln(\sigma_{q})-(\|\bs{\delta}_F\|^2+\|\bs{\delta}_C\|^2)\ln(\sigma_{\phi})-\|\bs{\delta}_p\|^2\ln(\sigma_{p})\}$.

\section{Relaxed GLRT}\label{relaxation_method}

Due to the optimization problems in (\ref{log_GLRT_test}) being nonconvex, it is difficult to obtain optimal solutions in general. We can apply the powerful tool of semidefinite programming relaxation (SDR) \cite{luo2010semidefinite} to change the optimization problems involving continuous variables in (\ref{log_GLRT_test}) to convex problems. This yields a new test. In this section we show the new test achieves a performance which is the same as the performance of the GLRT test in (\ref{log_GLRT_test}) for the typical values of $ {T_a} $ of interest. We start by defining some quantities to describe the new test.

Define $ S = N+L_F+L_C+1. $ Let $ \bs{X} $ denote an $ S\times S $ symmetric positive semidefinite matrix and recall that $ \bs{A} $ denotes an $ L\times N $ matrix of the form defined after (\ref{incidence_matrix}). All the other variables employed in the following functions were previously defined and are fixed by the physical problem we are solving. 
Define the functions 
\begin{align}
\tilde{g}(\bs{A},\bs{X},\bs{v}) &= C-\Tr{\frac{1}{2}\bs{M}({\bs{A}},\bs{v})\bs{X}},\label{def_tilde_g}
\end{align}
and
\begin{align}
\tilde{f}(m,{\bs{A}},\bs{X}) &= \Tr{\bs{Z}(m,{\bs{A}})\bs{X}}+\bs{b}_mp_0^2,\quad m=1,\dots,L,  \label{def_tilde_f}
\end{align}
where $ \Tr{\bullet} $ denotes the trace operator and 
\begin{align}
\bs{M}({\bs{A}},\bs{v}) &= \begin{bmatrix}
\bs{M}'_{11} & \bs{0} & \bs{M}'_{13}\\
\bs{0} & \bs{M}'_{22} & \bs{M}'_{23}\\
\bs{M}^{'T}_{13} & \bs{M}_{23}^{'T} & \bs{M}'_{33}\\
\end{bmatrix} \label{def_M_S0},
\end{align}
with
\begin{align}
\bs{M}'_{11} &= {T_a}\fracPS{p}\hpt\diag{\bs{\delta}_p},\\
\bs{M}'_{14} &= -\sum_{t=1}^{T_a}\fracPS{p}\tilde{\bs{p}}[t]\odot\bs{\delta}_p,\\
\bs{M}'_{22} &= {T_a}\fracPS{q}\bs{A}\hpt\diag{\bs{\delta}_q}\hpt\bs{A}^{T} + {T_a}\fracPS{\phi}\hpt\diag{\bs{\delta}_\phi},\\
\bs{M}'_{23} &= -\sum_{t=1}^{T_a}\fracPS{q}\bs{A}(\tilde{\bs{q}}[t]\odot\bs{\delta}_q) - \sum_{t=1}^{T_a}\fracPS{\phi}\bs{\delta}_\phi\odot\tilde{\bs{\phi}}[t], 
\end{align}
and
\begin{align}
\bs{M}'_{33} &= \sum_{t=1}^{T_a}\left(\fracPS{q}\left\| \tilde{\bs{q}}[t]\odot\bs{\delta}_q \right\|^2 + \fracPS{p}\left\| \tilde{\bs{p}}[t]\odot\bs{\delta}_p \right\|^2 + 
\right.\notag\\&\left. 
\fracPS{\phi}\left\| \bs{\delta}_\phi\odot\tilde{\bs{\phi}}[t] \right\|^2\right)\label{M_S_sub_m_9}.
\end{align}
In (\ref{def_tilde_f}), $ \bs{Z}(m,{\bs{A}}) = \diag{\begin{bmatrix} \bs{B}_{[m,:]}&-\bs{c}_m\bs{e}_{m,{1}\times L}&0 \end{bmatrix}} $ where $ \bs{B}_{[m,:]} $ denotes the $ m^{th} $ row of $ \bs{B} $, and $ \bs{e}_{m,1\times L} $ is an $ 1\times L $ row vector with all zero entries except for the $ m^{\text{th}} $ which is 1.
Let 
\begin{align}
\Omega({\bs{A}}) = &\left\{ \bs{X}\in\mathbb{R}^{S\times S}\hpt|\hpt\bs{X}\ge 0;\hspace{1em} \tilde{f}\left(m,{\bs{A}},\bs{X}\right) = 0;\hspace{1em}\bs{X}_{SS}= 1;\right.\notag\\
&\left.\bs{X}^T = \bs{X};\hspace{1em}m\in\{1,\cdots,L\} \right\} \label{def_Omega}
\end{align}
define a set of values of $ \bs{X} $. Here $ \bs{X}\ge 0 $ implies that $ \bs{X} $ is positive semidefinite. 
Let $ \mathcal{A}_{H_1} = \left\{ \bs{A}\hpt|\hpt\bs{A_{}}\ne\bs{A}_{H_0} \right\} $. 
\begin{theorem}[Standard GLRT Equivalent Expression] \label{Stnd_GLRT_THM}
	Let $ \bs{v} = ( \tilde{\bs{p}}[t],\tilde{\bs{\phi}}_F[t], \tilde{\bs{\phi}}_C[t],\tilde{\bs{q}}[t] ) $ represent the sensor observations. The GLRT in (\ref{log_GLRT_test}) can be expressed as 
	\begin{align}
	\sup_{\begin{subarray}{c} \bs{A} \in \mathcal{A}_{H_1}\\\bs{X}\in\Omega({\bs{A}})\\\rank{\bs{X}} = 1\end{subarray}}&\tilde{g}({\bs{A}},\bs{X},\bs{v}) - \sup_{\begin{subarray}{c} \\\bs{X}\in\Omega\left({\bs{A}_{H_0}}\right)\\\rank{\bs{X}} = 1\end{subarray}}\tilde{g}({\bs{A}_{H_0}},\bs{X},\bs{v}) \LRT{H_1}{H_0}\ln\rho. \label{equ_log_GLRT_test}
	\end{align}
\end{theorem}
\begin{proof}
	The proof is omitted as it involves only algebra. 
\end{proof}

The Relaxed GLRT is defined identically to (\ref{equ_log_GLRT_test}) but with the  $ \rank{\bs{X}} = 1 $ constraint removed.
\begin{definition}[Relaxed GLRT] \label{Relaxed_GLRT_THM}
	Let $ \bs{v} = ( \tilde{\bs{p}}[t],\tilde{\bs{\phi}}_F[t], \tilde{\bs{\phi}}_C[t],\tilde{\bs{q}}[t] ) $ represent the sensor observations. We define the Relaxed GLRT as the test 
	\begin{align}
	\sup_{\begin{subarray}{c} \bs{A}
		\in\mathcal{A}_{H_1}\\\bs{X}\in\Omega({\bs{A}})\end{subarray}}&\tilde{g}({\bs{A}},\bs{X},\bs{v}) - \sup_{\begin{subarray}{c} \\\bs{X}\in\Omega\left({\bs{A}_{H_0}}\right)\end{subarray}}\tilde{g}({\bs{A}_{H_0}},\bs{X},\bs{v}) \LRT{H_1}{H_0}\ln\rho. \label{relaxed_GLRT_test}
	\end{align}
\end{definition} 

In the sequel, the test in (\ref{log_GLRT_test}) is called the Standard GLRT. The advantage of the Relaxed GLRT is that it can be computed in a numerically reliable and efficient manner since the constraints for the continuous variables are convex. This is not true for the Standard GLRT in general due to the nonconvex optimization problems involved. 

One important property of the Relaxed GLRT is that it is equivalent to the Standard GLRT for a sufficiently large $ {T_a} $, i.e., the optimizations in \eqref{relaxed_GLRT_test} and \eqref{log_GLRT_test} yield identical optimal solutions. 
 Define $ \bs{\omega} = (\bs{p}^T,\hpt\bs{\phi}_F^T,\hpt\bs{\phi}_C^T)^T $ as the continuous part of $ \bs{\theta} $.
We define $ \mu(\bs{A}) $ and $ \xi(\bs{A}) $ as
\begin{align}
\mu(\bs{A}) &= \sup_{\bs{\omega}:\bs{f}(\bs{A},\bs{\omega}) = \bs{0}} \ln g(\bs{v}|\bs{\omega},\bs{A}),\label{original_problem}
\end{align}
and
\begin{align}
\xi(\bs{A}) &= \sup_{\begin{subarray}{c} \bs{X}\in\Omega({\bs{A}})\end{subarray}}\tilde{g}({\bs{A}},\bs{X},\bs{v}).\label{relaxed_problem}
\end{align}
Note that \eqref{original_problem} can represent either of the terms in \eqref{log_GLRT_test} which are both constrained maximum likelihood (ML) optimizations. Similarly, 
\eqref{relaxed_problem} can stand for either (relaxed) term in \eqref{relaxed_GLRT_test}. In the sequel, we call \eqref{relaxed_problem} the relaxed ML while \eqref{original_problem} is called the original ML.
From (13)-(15), define 
\begin{align}
\bs{n}_p[t] = \tilde{\bs{p}}[t] - {\bs{p}}[t]
\sim\mathcal{N}\left( \bs{0}, \sigma_p^2\bs{I}_N \right),\\	
\bs{n}_q[t] = \tilde{\bs{q}}[t] - {\bs{q}}[t]
\sim\mathcal{N}\left( \bs{0}, \sigma_q^2\bs{I}_N \right),
\end{align}
and
\begin{align}
\bs{n}_\phi[t] = \tilde{\bs{\phi}}[t] - {\bs{\phi}}[t] 
\sim\mathcal{N}\left( \bs{0}, \sigma_\phi^2\bs{I}_L \right).
\end{align}
The relaxed ML \eqref{relaxed_problem} is equivalent to the original one \eqref{original_problem}, i.e., $ \mu(\bs{A}) = \xi(\bs{A}) $, for sufficently large $ T_a$. 

\begin{theorem}[Asymptotic equivalence of the relaxed ML to the original ML] \label{thm_relaxed_is_exact}
	Suppose that the number of observations $ {T_a} $ satisfies
	\begin{align}
	{T_a} &> \max\left( \left\| \diag{\bs{\delta}_p}\sum_{t=1}^{T_a}\bs{n}_p[t] \right\|_{\infty}, \hspace{3pt}\left\| \diag{\bs{\delta}_\phi}\sum_{t=1}^{T_a}\bs{n}_\phi[t] \right\|_{\infty},\right.\notag\\&\left.\hspace{3pt}\left\| \bs{A}\hspace{2pt}\diag{\bs{\delta}_q}\sum_{t=1}^{T_a}\bs{n}_q[t] \right\|_{\infty}/\eta_{\min}\left[ \bs{A}\hspace{2pt}\diag{\bs{\delta}_q}\hspace{2pt}\bs{A}^T \right] \right), \label{exact_relaxed_final_condi}
	\end{align}
	where $ \eta_{min}[\bs{A}\hspace{2pt}\diag{\bs{\delta}_q}\hspace{2pt}\bs{A}^T] $ denotes the smallest non-zero eigenvalue of $ \hspace{2pt}\bs{A}\hspace{2pt}\diag{\bs{\delta}_q}\hspace{2pt}\bs{A}^T $, and $ \|(x_1\allowbreak,x_2,\allowbreak\dots,\allowbreak x_m)\|_{\infty} = \max(|x_1|,|x_2|,\dots,|x_m|) $ denotes the infinity norm. Then, the relaxed ML \eqref{relaxed_problem} is equivalent to the original ML \eqref{original_problem}, i.e., $ \mu(\bs{A}) = \xi(\bs{A}) $. 
This implies the Relaxed GLRT is equivalent to the Standard GLRT when $ {T_a} $ satisfies \eqref{exact_relaxed_final_condi} for all possible $ \bs{A} $. 
\end{theorem}
\begin{proof}
	The proof is shown in Appendix \ref{app:proofOfTheoremAsympEquivalent}.
\end{proof}

The essence of our proof shows our constrained optimization problem is a member of a more general class of problems considered in Proposition 5 in \cite{so2010probabilistic}. We can also give an intuitive justification to Theorem \ref{thm_relaxed_is_exact}. 
When we solve each optimization in \eqref{relaxed_GLRT_test}  (the relaxed problem) for a sufficiently large $ T_a $, the rank of the optimal solution for $ \bs{X} $ is always $ 1 $. 
From (26), it can be seen that the deterministic components of $  \bs{M}(\bs{A},\bs{v})  $ will be much larger than the random components for a sufficiently large $ T_a $. The dominance becomes even more significant as $ T_a $ is increased. By ignoring the much smaller random parts of $ \bs{M}(\bs{A},\bs{v}) $, the solution to each optimization in \eqref{relaxed_GLRT_test}  (the relaxed problem) is easily shown to be unit rank.
By comparing Theorem 1 and Definition 1, the rank-1 solutions are exactly the optimal solutions to each optimization in \eqref{equ_log_GLRT_test} (the unrelaxed standard problem). Thus the relaxed problem gives the solution for the unrelaxed problem for a sufficiently large $ T_a $. In Section VI, we provide numerical results which show the rank of $ \bs{X} $ is $ 1 $ for a sufficiently large $ T_a $ for an example.

\newcommand{\caltilG}{\mathcal{\tilde{G}}}

\subsection{An Efficient Topology Verification Algorithm} \label{eq:effiTopoVerAlgo}

The following algorithm is a more efficient way to solve \eqref{def_Set_H1} as opposed to employing \eqref{relaxed_GLRT_test}. Define $ \caltilG = \sup_{\begin{subarray}{c} \\\bs{X}\in\Omega\left({\bs{A}_{H_0}}\right),\hpt\rank{\bs{X}} = 1\end{subarray}}\tilde{g}({\bs{A}_{H_0}},\bs{X},\bs{v}) $.

{\bf Step 1: } First check if the observed data $ \bs{v} $ produces a sufficiently small value of $ \caltilG $, which implies the data does not fit the assumed model if $ H_0 $ is true, so $ H_0 $ is not true. This allows us to decide $ H_1 $ must be true without further calculations. If $ \caltilG $ is small, then from \eqref{def_tilde_g}, $ \inf_{\begin{subarray}{c} \\\bs{X}\in\Omega\left({\bs{A}_{H_0}}\right)\\\rank{\bs{X}} = 1\end{subarray}}\Tr{\frac{1}{2}\bs{M}({\bs{A}},\bs{v})\bs{X}} $ is larger than zero. Thus, sufficiently small $ \caltilG $ implies
\begin{align}
	\inf_{\begin{subarray}{c} \\\bs{X}\in\Omega\left({\bs{A}_{H_0}}\right)\\\rank{\bs{X}} = 1\end{subarray}}\Tr{\frac{1}{2}\bs{M}({\bs{A}},\bs{v})\bs{X}} > \epsilon, \label{eq:testFitness}
\end{align}
for a suitable positive $ \epsilon $. Thus we can use \eqref{eq:testFitness} to decide $ H_1 $ without further calculations.
It is easily shown from Theorem \ref{thm_relaxed_is_exact} that if $ H_0 $ is true, the left hand side of \eqref{eq:testFitness} will approach $ 0 $ as $ T_a\rightarrow\infty $. Thus \eqref{eq:testFitness} will not ever be true if $ T_a $ is sufficiently large and $ H_0 $ is true. Thus if \eqref{eq:testFitness} is true, we stop and decide $ H_1 $ is true.

{\bf Step 2: } If \eqref{eq:testFitness} is not true, then we employ \eqref{relaxed_GLRT_test}. We set $ \bs{A} = \bs{A}_{H_0} $ and employ the gradient guided search in \cite{hu2000gradient} to solve for $ \bs{A} $ in the first term in \eqref{relaxed_GLRT_test}. The gradient guided search will first try all $ \bs{A} $ which differ from $ \bs{A}_{H_0} $ by a one link change. It will pick the one link change that increases the first term in \eqref{relaxed_GLRT_test} the most. If none of the one link changes increase the first term in \eqref{relaxed_GLRT_test}, it will stop and evaluate \eqref{relaxed_GLRT_test}. If it finds a one link change that increases the first term of \eqref{relaxed_GLRT_test}, it picks the one link change with the largest increase and repeats this procedure. If all possible topologies have been tested, the algorithm will stop and evaluate \eqref{relaxed_GLRT_test}. 

Now we explain why the algorithm is efficient. If the topology is changed significantly (true $ \bs{A} $ very different from $ \bs{A}_{H_0} $), then the algorithm will stop during step 1 and no search over the topology $ \bs{A} $ is needed. This will occur when the attacker tries to cause significant changes which is typically going to cause significant problems if the operator uses this incorrect network in planning. If the topology undergoes only small changes, then the gradient guided search will find the correct topology quickly. Thus, with appropriate $ \epsilon $, the algorithm will always stop quickly. We illustrate with numerical results this and discuss how to pick an appropriate threshold $ \epsilon $ in Section \ref{sec_num_simulation}.

We end this section by describing a minimum requirement on the sensor placement. 
\begin{definition}[Proper Sensor Placement]\label{lemma_app_sensor_placement}
	A minimum requirement for a proper sensor placement for the algorithms given in Theorem \ref{Stnd_GLRT_THM} and Definition \ref{Relaxed_GLRT_THM} is that the placement ensures the constrained maximum likelihood estimate of $ \bs{\theta} $ is consistent (converges to the correct value as $ {T_a}\rightarrow\infty $) under hypothesis $ H_i,i\in\{0,1\} $ when the data used for the estimation is from $ H_i,i\in\{0,1\}$.
\end{definition}
In the sequel, we consider only sensor placements which satisfy {\bf Definition \ref{lemma_app_sensor_placement}}. At the end of the next section, we provide a direct calculation to make sure {\bf Definition} 2 is satisfied.

\section{Asymptotic Performance Approximation}\label{sec_asym_perf}

In this section, we study the asymptotic performances for the Standard GLRT and the Relaxed GLRT.  
Define the Fisher Information Matrix (FIM) for estimating $ \bs{\omega} $ given a general $ \bs{A} $ as
\begin{align}
\bs{J}(\bs{A},\bs{\omega})=E_{\bs{\omega}}\left[\frac{\partial}{\partial\bs{\omega}^T}\ln g(\bs{v}|\bs{A},\bs{\omega})\frac{\partial}{\partial\bs{\omega}}\ln g(\bs{v}|\bs{A},\bs{\omega}) \right] \label{def_FIM},
\end{align} 
where $ E_{\bs{\omega}}(\bullet) $ denotes the expected value computed by averaging using the pdf in (\ref{e14})  for the assumed value of $ \bs{\omega} $. Let $ \bs{F}(\bs{A},\bs{\omega}) = \frac{\partial}{\partial\bs{\omega}}\bs{f}(\bs{A},\bs{\omega}) $ where $ \bs{f}(\bs{A},\bs{\omega}) $ was defined in \eqref{con_func_stan_GLRT}. Then, for a general $ \bs{A} $ define the matrix $ \bs{U}(\bs{A},\bs{\omega})\in\mathbb{R}^{(N+L)\times \rank{\bs{F}(\bs{A},\bs{\omega})}} $ as the solution to 
\begin{align}
\bs{F}(\bs{A},\bs{\omega})\bs{U}(\bs{A},\bs{\omega})  = \bs{0} \hspace{7pt}\text{such that}\hspace{7pt} \bs{U}^T(\bs{A},\bs{\omega})\bs{U}(\bs{A},\bs{\omega}) = \bs{I}. \label{def_U}
\end{align}
Using these definitions, we define the constrained Cramer Rao Bound as 
\begin{align}
\text{CCRB}(\bs{A},\bs{\omega})=&\bs{U}(\bs{A},\bs{\omega})[\bs{U}^T(\bs{A},\bs{\omega})\notag\\&\hpt\bs{J}(\bs{A},\bs{\omega})\hpt\bs{U}(\bs{A},\bs{\omega})]^{-1}\bs{U}^T(\bs{A},\bs{\omega})\label{def_CCRB}.
\end{align}

\begin{lemma}\label{lemma_choose_wrong_A_go_zero_asym}
	Let $ {\bs{\theta}}^* = ( \vecwang{\bs{A}^*}, \bs{\omega}^* )^T $ be the solution to 
	$ sup_{\bs{\theta}\in\Theta_{H_1}} \ln g(\bs{v}|\bs{\theta}) $ when the data $ \bs{v} $ is obtained under 
	$ H_1 $ from a network with a topology described by $ \bs{A}_{H_1} $ with the noise free network variables $ \bs{\omega}_{H_1} = (\bs{p}^T_{H_1},\hpt\bs{\phi}_{H_1,F}^T,\hpt\bs{\phi}_{H_1,C}^T)^T $.  Given {\bf Definition}~\ref{lemma_app_sensor_placement}, $ {\bs{\theta}}^* $ approaches $ ( \vecwang{\bs{A}_{H_1}}, \bs{\omega}_{H_1} )^T $ as $ {T_a}\rightarrow\infty $.  This implies that the probability that $ \bs{A}^* \ne\bs{A}_{H_1} $ goes to 0 as $ {T_a}\rightarrow\infty $.
\end{lemma}

The proof follows from the definition of consistency.   Thus for $ {T_a}\rightarrow\infty $, all the estimates must be exactly the true values in the system.  

\begin{lemma}\label{lemma_new}
	For the case in Lemma~\ref{lemma_choose_wrong_A_go_zero_asym}, for a sufficiently large but finite $T_a$, the probability that $ \bs{A}^* \ne\bs{A}_{H_1} $ can be made arbitrary close to zero. 
\end{lemma}

\begin{proof}

Based on the Karush Kuhn Tucker (KKT) conditions we can show that if $ \bs{A}=\bs{A}_{H_1} $, all the elements of $ \bs{\omega}^* $ are unbiased and their variances are inversely proportional to $ {T_a} $. Thus considering $ \bs{\phi}^* $ when $ \bs{A}=\bs{A}_{H_1} $, we can show 
\begin{align}
	E\{ &\bs{\phi}^* \}=[\sigma_\phi^2\bs{A}\hpt\diag{\bs{\delta}_q}\hpt\bs{A}^T + \sigma_q^2\hpt\diag{\bs{\delta}_\phi}]^\dagger\notag\\	&[\sigma_\phi^2\bs{A}\hpt\diag{\bs{\delta}_q}\hpt\bs{A}+\sigma_q^2\hpt\diag{\bs{\delta}_\phi}]\bs{\phi}_{H_1}=\bs{\phi}_{H_1},\label{mean_unbiased}
\end{align}
and
\begin{align}
\text{Cov}[\bs{\phi}^*] 
=\frac{\sigma_\phi^2\sigma_q^2}{T_a}[\sigma_\phi^2\bs{A}\diag{\bs{\delta}_q}\bs{A}^T + \sigma_q^2\diag{\bs{\delta}_\phi}]^{\dagger T}. \label{cov_unbiased}
\end{align}
We can show similar equations for the other elements of $ \bs{\omega}^* $.
Now if $ \bs{A}\ne\bs{A}_{H_1} $, the KKT conditions show that all the elements of $ \bs{\omega}^* $ are biased and still their variance are inversely proportional to $ {T_a} $. 

When $ \bs{A} = \bs{A}_{H_1} $, all the scalar parameter estimates of the components of $ \bs{\omega} $, other than $ \bs{A} $, are unbiased and the variance of each parameter estimate becomes smaller as ${T_a}$ increases.  Thus, the parameter estimates get better and better as ${T_a}$ increases since their distribution becomes more closely focused around the correct value. This makes all the squared norm terms in $ \ln g(\bs{v}|\bs{\theta}) $ in \eqref{likelihood_f} very close to zero so the likelihood $ g(\bs{v}|\bs{\theta}) $ is as large as it can be.  Note that the squared norm term involving $\tilde{\bs{q}}[t]$ in $ \ln g(\bs{v}|\bs{\theta}) $ is also a narrow Gaussian centered on zero (becoming more narrow as ${T_a}$ increases) so the argument to this square norm is also very close to zero. This is what happens if $ \bs{A} = \bs{A}_{H_1} $.

Now if $ \bs{A}\ne\bs{A}_{H_1} $, then the mean vector of the estimate of $ \bs{\omega} $ will be biased and the arguments to the squared norm terms of $ \ln g(\bs{v}|\bs{\theta}) $ in \eqref{likelihood_f} are also biased and again these terms are Gaussian.  Now as ${T_a}$ is increased the variances of each of these terms will decrease so the distribution becomes more closely focused around an incorrect value.  At a large enough ${T_a}$, the squared error in these term with bias ($ \bs{A} \ne \bs{A}_{H_1} $) will be be much larger than the unbiased error in the estimates when $ \bs{A} = \bs{A}_{H_1} $ and all the variances are small.  In fact at large enough ${T_a}$, the overlap between the probability density functions of the biased terms and the unbiased terms will be small and we can make this overlap as small as we like by increasing ${T_a}$.  This makes the probability that the likelihood with $ \bs{A} = \bs{A}_{H_1} $ is larger than that with $ \bs{A} \ne \bs{A}_{H_1} $ true with a probability near unity which we can make as close as we like to unity by increasing ${T_a}$. This completes the proof.
\end{proof} 


Based on Lemma~\ref{lemma_new}, under $H_1$ we must have $ \bs{A}^* = \bs{A}_{H_1} $ with high probability for sufficiently large ${T_a}$ so (\ref{log_GLRT_test}) is asymptotically equivalent with high probability to 
\begin{align}
\sup_{\bs{\omega}\in\mathcal{A}_{\omega_1}}\ln g(\bs{v}|\bs{A}_{H_1},\bs{\omega}) - \sup_{\bs{\omega}\in\mathcal{A}_{\omega_0}}\ln g(\bs{v}|\bs{A}_{H_0},\bs{\omega}) \LRT{H_1}{H_0}\ln\rho. \label{equ_GLRT_test}
\end{align}
in which $ \mathcal{A}_{\omega_i} = \left\{ \bs{\omega}: f(\bs{A}_{H_i}, \bs{\omega}) = \bs{0} \right\},i\in\{0,1\} $.  Further simplification results by using the fact that such a test is asymptotically equivalent to the corresponding Wald test as described in the following theorem. 
\begin{theorem}\label{GLRT_App_Wald_THM}
	Under $H_1$, for $ {T_a}\rightarrow\infty $ the test in (\ref{equ_GLRT_test}) becomes equivalent to the test 
	\begin{align}
	\left[\hat{\bs{\omega}}-\bs{\omega}^*\left(\bs{A}_{H_0}\right)\right]^T&\text{CCRB}^\dagger\left(\bs{A}_{H_0},\bs{\omega}^*\left(\bs{A}_{H_0}\right)\right)\notag\\&\left[\hat{\bs{\omega}}-\bs{\omega}^*\left(\bs{A}_{H_0}\right)\right]\LRT{H_1}{H_0}\ln\rho, \label{cons_wald_test}
	\end{align}
	where $ \text{CCRB}^\dagger\left(\bs{A}_{H_0},\bs{\omega}^*\left(\bs{A}_{H_0}\right)\right) $ denotes the pseudo-inverse of (\ref{def_CCRB}), $ \hat{\bs{\omega}} = \arg\sup_{\omega\in\mathcal{A}_{\omega_1}}\ln g(\bs{v}|\bs{A}_{H_1},\bs{\omega}) $, and $ \rho $ is chosen to fix the false alarm probability.
\end{theorem}
\begin{proof}
	The proof follows from the result in \cite{moore2010constrained} combined with Appendix 6A in \cite{kay1998fundamentals}. 
\end{proof}

The following lemma describes the distribution of (\ref{cons_wald_test}).

\begin{lemma}[Distribution of (\ref{cons_wald_test})]\label{C_chi_square_lemma}
	Define 
	$ \lambda \allowbreak= \allowbreak[\bs{\omega}^*(\bs{A}_{H_1}) \allowbreak-\allowbreak \bs{\omega}^*\left(\bs{A}_{H_0}\right)]^T\allowbreak\text{CCRB}^\dagger(\bs{A}_{H_0},\allowbreak \bs{\omega}^*(\bs{A}_{H_1}))\allowbreak\left[\bs{\omega}^*(\bs{A}_{H_1}) - \bs{\omega}^*\left(\bs{A}_{H_0}\right)\right] $. Then, under $ H_1 $ the test statistic in (\ref{cons_wald_test}) asymptotically (large $ {T_a} $) follows the non-central chi-squared distribution with $ L+N $ degrees of the freedom and non-centrality parameter $ \lambda $.
\end{lemma}
\begin{proof}
	The proof is shown in Appendix 6C in \cite{kay1998fundamentals}.
\end{proof}

Define the Marcum Q function with $ w $ degrees of freedom \cite{sun2010monotonicity} as
\begin{align}
Q_w(a,b) = \int_{b}^{\infty} x\left( \frac{x}{a} \right)^{w-1}\exp\left( -\frac{x^2+a^2}{2} \right)I_{w-1}(ax)dx\label{eq:defMarcumQ}
\end{align}
where $ I_{w-1}(x) $ is the modified Bessel function of order $ w-1 $. Then the inverse of the Marcum Q function $ Q^{-1}_w(c,d) $ is defined as the solution $ y $ to $ Q_w(c,y) = d$ \cite{gil2014asymptotic}. The following theorem employs Lemma \ref{C_chi_square_lemma} to describe the performance of the Standard GLRT for large $ T $.

\begin{theorem}[Asymptotic Performance Approximation of the Standard GLRT] \label{Asymptotic_Performance_of_Standard_GLRT}
	For a sufficiently large $ {T_a} $ and a fixed false alarm probability $ P_{fa} $, an accurate expression for the detection probability for the Standard GLRT is 
	\begin{align}
	P_d = Q_{L+N}(\sqrt{\lambda}, \sqrt{\rho}) \label{Pd_standard}
	\end{align}
	where $ \rho = Q_{L+N}^{-1}(0, P_{fa}) $, $ \bs{U} = \bs{U}(\bs{A}_{H_0},\bs{\omega}^*\left(\bs{A}_{H_0}\right)) $ was defined in (\ref{def_U}),
	\begin{align}
	\lambda = \left[\bs{\omega}^*(\bs{A}_{H_1}) - \bs{\omega}^*\left(\bs{A}_{H_0}\right)\right]^T&\bs{U} \bs{U}^T\bs{J}\left(\bs{A}_{H_0}\right)\bs{U} \bs{U}^T\notag\\&\left[\bs{\omega}^*(\bs{A}_{H_1}) - \bs{\omega}^*(\bs{A}_{H_0})\right]\label{def_lambda},
	\end{align}
	and (\ref{def_FIM}) yields
	\begin{align}
	\bs{J}\left(\bs{A}_{H_0}\right) =  {T_a}\begin{bmatrix}
	\bs{J}'_{pp} & \bs{0} \\
	\bs{0} & \bs{J}'_{\phi\phi}\left(\bs{A}_{H_0}\right)
	\end{bmatrix}\label{fisher_infor_m_all},
	\end{align}
	where the sub-matrices of $ \bs{J}\left(\bs{A}_{H_0}\right) $ are
	\begin{align}
	\bs{J}'_{pp} &= \fracPS{p}\diag{\bs{\delta}_p},\\
	\bs{J}'_{\phi\phi}\left(\bs{A}_{H_0}\right) &= \frac{1}{\sigma_{\phi}^2}\diag{\bs{\delta}_\phi} + \frac{1}{\sigma_q^2}\bs{A}_{H_0}\diag{\bs{\delta}_q}\bs{A}_{H_0}^T.\label{fisher_infor_standard_GLRT_sub_m_end} 
	\end{align}
	Note that $ P_d $ in (\ref{Pd_standard}) increases monotonically with $ {T_a} $. 
\end{theorem}
\begin{proof}
	Using Lemma \ref{C_chi_square_lemma}, under $ H_1 $ the test in (\ref{log_GLRT_test}) asymptotically follows the non-central chi-squared distribution with $ L+N $ degrees of freedom and the non-centrality parameter
	\begin{align}
	\lambda = \left[\bs{\omega}^*(\bs{A}_{H_1}) - \bs{\omega}^*(\bs{A}_{H_0})\right]^T&\left\{\bs{U}\left[ \bs{U}^T\bs{J}(\bs{A}_{H_0})\bs{U} \right]^{-1}\bs{U}^T\right\}^\dagger\notag\\&\left[\bs{\omega}^*(\bs{A}_{H_1}) - \bs{\omega}^*(\bs{A}_{H_0})\right],
	\end{align}
	To resolve the pseudo inverse, note that $ (\bs{E}\bs{E}^T)^\dagger=\bs{E}(\bs{E}^T\bs{E})^{-2}\bs{E}^T $.
	Let $ \bs{E} = \bs{U}\left[\bs{U}^T\bs{J}(\bs{A}_{H_0})\bs{U}\right]^{-1/2} $ then we have 
	\begin{align}
	\lambda = \left[\bs{\omega}^*(\bs{A}_{H_1}) - \bs{\omega}^*(\bs{A}_{H_0})\right]^T&\bs{U} \bs{U}^T\bs{J}(\bs{A}_{H_0})\bs{U}\bs{U}^T\notag\\&\left[\bs{\omega}^*(\bs{A}_{H_1}) - \bs{\omega}^*(\bs{A}_{H_0})\right].\label{proof_final_lambda}
	\end{align}
\end{proof}
From Theorem \ref{thm_relaxed_is_exact}, the Relaxed GLRT is asymptotically equivalent to the Standard GLRT. 
Thus to derive the asymptotic performance of the Relaxed GLRT, we just need to calculate the $ \bs{\omega}^*(\bs{A}) $ corresponding to the optimal $ \bs{X}^*(\bs{A}) $ obtained from \eqref{relaxed_problem}. Using (\ref{def_Omega}) and equating the solution of (\ref{equ_log_GLRT_test}) to (\ref{relaxed_GLRT_test}) (with significant algebra) when \eqref{exact_relaxed_final_condi} is satisfied 
\begin{align}
\bs{X}^*(\bs{A}) &= \begin{bmatrix}
\bs{\omega}^*(\bs{A}) \\ 1
\end{bmatrix}\begin{bmatrix}
\bs{\omega}^{*T}(\bs{A}) & 1
\end{bmatrix}.
\end{align}
Given the definitions of $S$ and $\bs{X}$ before (26), the solution, called $ \overline{\bs{\omega}}^*(\bs{A}) $, is  
\begin{align}
\begin{bmatrix}
\overline{\bs{\omega}}^*(\bs{A}) \\ 1
\end{bmatrix} = \bs{X}^*_{[:,S]}(\bs{A}).
\end{align}
in which $ \bs{X}^*_{[:,S]}(\bs{A}) $ is $ S^{th} $ column of $ \bs{X}^*(\bs{A}) $. 
Therefore, the asymptotic performance of the Relaxed GLRT under $H_1$ can be obtained by substituting $ \overline{\bs{\omega}}^*(\bs{A}_{H_1}) $ into the asymptotic performance of the Standard GLRT  shown in Theorem \ref{Asymptotic_Performance_of_Standard_GLRT}.

\newcommand{\Aoepn}{\bs{A}_{\text{open}}}
Next we provide a simplied sensor placement procedure. 
 Let $ \Aoepn $ denote the incidence matrix as per \eqref{incidence_matrix} for the topology with all changeable pipelines opened. Similarly, we have $ \bs{B}_{\text{open}} $ defined using the conventions in \eqref{def_B}. 
\begin{theorem}[Simplified Proper Sensor Placement]\label{thm_suff_cond_app_sen_pla}
	Define $ \widetilde{\bs{F}}(\bs{A}) = [\bs{B}\quad \diag{\bs{c}}] $. Let $ \widetilde{\bs{U}}(\bs{A}) = [\bs{I}_{N}\quad-\bs{B}^T\text{diag}^{-1}(\bs{c})]^T $.
	The requirement of Definition \ref{lemma_app_sensor_placement} will be guaranteed if 
	\begin{align}
	&\text{rank}\left[\widetilde{\bs{U}}^T(\bs{A}_{\text{open}})\bs{J}\left(\bs{A}_{\text{open}}\right)\widetilde{\bs{U}}(\bs{A}_{\text{open}})\right] = N. \label{eq:app_pre_condition}
	\end{align}
\end{theorem} 
\begin{proof}
	Based on Theorem 2 in \cite{moore2008maximum}, for a general $ \bs{A} $, to satisfy the requirement of Definition 2, we need
	\begin{align}
	&\text{rank}\left[\bs{U}^T\left(\bs{A},\bs{\omega}^*\left(\bs{A}\right)\right)\bs{J}\left(\bs{A}\right)\bs{U}\left(\bs{A},\bs{\omega}^*\left(\bs{A}\right)\right)\right] =\notag\\&	S-\text{rank}\left\{\bs{F}\left(\bs{A},\bs{\omega}^*\left(\bs{A}\right)\right)\right\}-1. \label{eq:app_pre_condition_with_w}
	\end{align}
	Recall that $ \bs{\omega}^* = [\bs{p}^{*T}\quad \bs{\phi}^{*T}]^T $.
	By taking the partial derivative of $ \bs{f}(\bs{A},\bs{\omega}) $ with respect to $ \bs{\omega} $, we have
	\begin{align}
	\bs{F}\left(\bs{A},\bs{\omega}^*\left(\bs{A}\right)\right) &= \begin{bmatrix}
	2\bs{B}\hpt\diag{\bs{p}^*} & 2\diag{\bs{c}\odot\bs{\phi}^*}
	\end{bmatrix} \notag\\
	&= \widetilde{\bs{F}}(\bs{A})\hpt\diag{2\bs{\omega}^*\left(\bs{A}\right)}.\label{eq:explicit_F}
	\end{align}
	where the ranks of $ \bs{F}\left(\bs{A},\bs{\omega}^*\left(\bs{A}\right)\right) $ and $ \widetilde{\bs{F}}(\bs{A}) $ are both $ L $ since $ \bs{c}_l > 0,\hpt l = 1,\dots,L $ as per \eqref{weymouth}. 
	
	Define a square non-singular matrix $ \bs{P}\in\mathbb{R}^{N\times N} $ such that $ \bs{U}\left(\bs{A},\bs{\omega}^*\left(\bs{A}\right)\right) = \widetilde{\bs{U}}(\bs{A})\bs{P} $. This matrix $ \bs{P} $ exists since the columns of both $ \bs{U}\left(\bs{A},\bs{\omega}^*\left(\bs{A}\right)\right) $ and $ \widetilde{\bs{U}}(\bs{A}) $ are individually a basis for the null space of $ \widetilde{\bs{F}}(\bs{A}) $. Then, we have
	\begin{align}
	&\text{rank}\left[\bs{U}^T\left(\bs{A},\bs{\omega}^*\left(\bs{A}\right)\right)\bs{J}\left(\bs{A}\right)\bs{U}\left(\bs{A},\bs{\omega}^*\left(\bs{A}\right)\right)\right] \notag\\&= \text{rank}\left[\bs{P}^T\widetilde{\bs{U}}(\bs{A})^T\bs{J}\left(\bs{A}\right)\widetilde{\bs{U}}(\bs{A})\bs{P}\right] \notag\\&= \text{rank}\left[\widetilde{\bs{U}}^T(\bs{A})\bs{J}\left(\bs{A}\right)\widetilde{\bs{U}}(\bs{A})\right]. \label{eq:rankUJUsam}
	\end{align}
	
	By \eqref{fisher_infor_m_all}-\eqref{fisher_infor_standard_GLRT_sub_m_end} and the definition of $ \widetilde{\bs{U}}(\bs{A}) $, we have
	\begin{align}
		&\widetilde{\bs{U}}^T(\bs{A})\bs{J}\left(\bs{A}\right)\widetilde{\bs{U}}(\bs{A}) =
		  \fracPS{\phi}\bs{B}^T\text{diag}^{-1}(\bs{c})\diag{\bs{\delta}_\phi}\text{diag}^{-1}(\bs{c})\hpt\bs{B}+\notag\\
		 &\fracPS{\phi}\bs{B}^T\text{diag}^{-1}(\bs{c})\bs{A}\hpt\diag{\bs{\delta}_q}\hpt\bs{A}^T\text{diag}^{-1}(\bs{c})\hpt\bs{B}+\frac{1}{\sigma_p^2}\diag{\bs{\delta}_p}.\label{eq:longEqRank}
	\end{align}
	The first term of the right hand side of \eqref{eq:longEqRank} is the weighted Laplacian matrix (edges are weighted by $ \bs{c} $ and $ \bs{\delta}_\phi $)
	while the 2nd term has the same rank as the weighted Laplacian matrix. Connecting any unconnected links will increase or maintain the rank of a Laplacian matrix (Section 2 in \cite{merris1994laplacian}). Thus, $ \text{rank}\left[\widetilde{\bs{U}}^T(\bs{A}_{\text{open}})\bs{J}\left(\bs{A}_{\text{open}}\right)\widetilde{\bs{U}}(\bs{A}_{\text{open}})\right] \le \text{rank}\left[\widetilde{\bs{U}}^T(\bs{A})\bs{J}\left(\bs{A}\right)\widetilde{\bs{U}}(\bs{A})\right] $. 
	Hence, if \eqref{eq:app_pre_condition} is true for $ \bs{A}_{\text{open}} $, then \eqref{eq:app_pre_condition} must be true for all $ \bs{A} $  obtained from $ \bs{A}_{\text{open}} $ by closing changeable pipelines.
\end{proof}

We first define an optimal placement as a placement which satisfies Definition 2 while having the lowest cost.
We can find the optimal placement by solving 
\begin{align}
\min_{\bs{\delta}} &\sum_{n=1}^N\left(d_{p,n}{\delta}_{p,n}+d_{q,n}{\delta}_{q,n}\right)+\sum_{l=1}^{L_F}d_{F,l}{\delta}_{F,l}+\sum_{k=1}^{L_C}d_{C,k}{\delta}_{C,k},\label{opt_sen_pla_equ1}\\
&\text{subject to: } (\ref{eq:app_pre_condition}), \label{constraints_in_conclusion}
\end{align}
where $ d_{p,n} $ and $ d_{q,n} $ are costs for placing pressure and injection sensors at the $ n^{th} $ node respectively, and $ d_{F,l} $ and $ d_{C,l} $ are costs for placing flow meters at the $ l^{th} $ pipeline, respectively. Solving \eqref{opt_sen_pla_equ1} and \eqref{constraints_in_conclusion} is hard since \eqref{eq:app_pre_condition} is not linear. 

Thus we describe a practical approach to find a placement which is  simple and typically provides acceptable cost (not so much more than optimum) based on our numerical investigations, some given in this paper. First note that placing a pressure sensor at every node always satisfies Definition \ref{lemma_app_sensor_placement}. This can be easily proved by substituting $ \bs{\delta}_p = \bs{1} $ into \eqref{eq:app_pre_condition}. Then, we can try to remove an arbitrary sensor from the placement and add another sensor which has a lower cost and allows a new placement which satisfies \eqref{eq:app_pre_condition}. We repeat this procedure until the cost cannot be further decreased. Based on \eqref{eq:app_pre_condition}, the minimum number of placed sensors is $ N $. Thus, our starting point $ \bs{\delta}_p = \bs{1} $ uses this minimum number of sensors. Since the replacing step only deceases the total cost, the final solution can only be closer to the optimal. This approach is efficient since it starts at a  point with the minimum number of sensors and trys to exchange sensors to reduce cost. We provide numerical results for this placement approach in Section \ref{sec_num_simulation}.

\section{Numerical Results}\label{sec_num_simulation}

\begin{figure}
	\centering
	\subfloat[]{\includegraphics[width=0.23\textwidth]{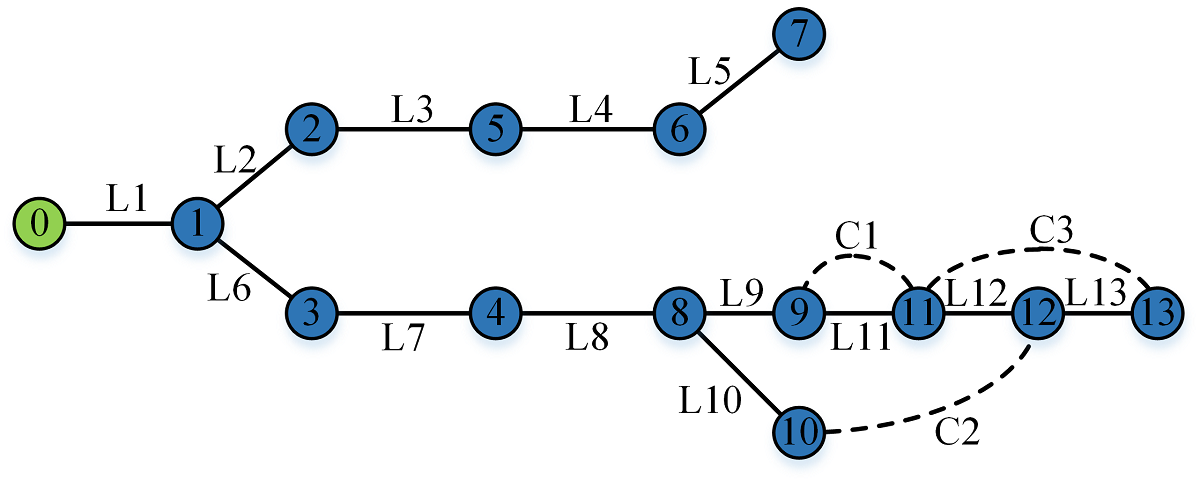}\label{fig_model_gas_network_Bining}}
	\hfil
	\subfloat[]{\includegraphics[width=0.23\textwidth]{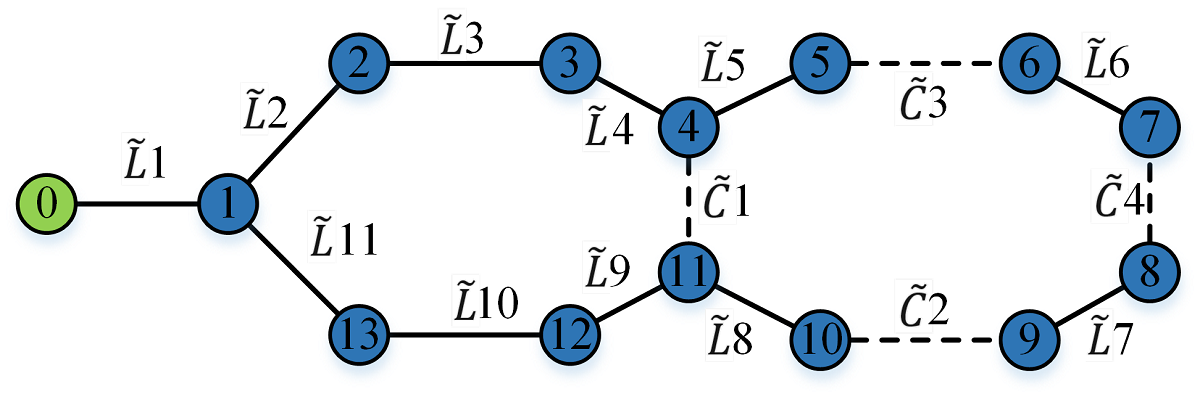}\label{fig_model_gas_network_PESGM17}}
	\caption{Natural gas networks evaluated in this paper. Dashed lines indicate changeable pipelines while solid lines are fixed pipelines. Blue nodes involve withdraws of gas from networks where as the green node involves the injection of gas into networks. (a) Network 1 \cite{zhao2017coordinated} (b) Network 2 \cite{ojha2016solving}.}
	\label{fig_model_gas_network}
\end{figure}

\begin{table}
	\renewcommand{\arraystretch}{1.3}
	\caption{Changeable pipeline settings under $ H_0 $/$ H_1 $ for 4 different cases in Network 1}
	\label{table_different_cases_net1}
	\centering
	\begin{tabularx}{0.43\textwidth}{c *3{>{\Centering}X}}
		\hline\hline
		& \multicolumn{3}{c}{ Changeable pipeline states ($ H_0 $/$ H_1 $)}\\
		\cline{2-4}
		{Case} & $ C1 $ & $ C2 $ & $ C3 $ \\
		\hline
		Case 1 & Closed/Closed & Closed/Closed & Closed/Open \\
		
		Case 2 & Open/Open & Closed/Closed & Open/Closed \\
		
		Case 3 & Closed/Closed & Closed/Closed & Open/Closed \\
		
		Case 4 & Open/Closed & Open/Closed & Open/Closed \\
		\hline\hline
	\end{tabularx}
\end{table}

\begin{table}
	\renewcommand{\arraystretch}{1.3}
	\caption{Changeable pipeline settings under $ H_0 $/$ H_1 $ for 2 different cases in Network 2}
	\label{table_different_cases_net2}
	\centering
	\begin{tabularx}{0.50\textwidth}{c*4{>{\Centering}X}}
		\hline\hline
		& \multicolumn{4}{c}{ Changeable pipeline states ($ H_0 $/$ H_1 $)}\\
		\cline{2-5}
		{Case} & $ \tilde{C}1 $ & $ \tilde{C}2 $ & $ \tilde{C}3 $ & $ \tilde{C}4 $\\
		\hline
		Case 1 & Closed/Closed & Closed/Closed & Closed/Closed & Closed/Open \\
		Case 2 & Closed/Open & Closed/Closed & Closed/Closed & Closed/Open \\
		\hline\hline
	\end{tabularx}
\end{table}

\begin{table}[!t]
	\renewcommand{\arraystretch}{1.3}
	\caption{Gas injections and pipeline characteristic values used in numerical investigations.}
	\label{table_variables}
	\centering
	\begin{tabular}{lllllllll}
		\hline\hline
		\multicolumn{4}{c}{Network 1} & & \multicolumn{4}{c}{Network 2}\\
		\cline{1-4} \cline{6-9}
		Node & $ \bs{q} $ & Pipe. & $ \bs{c} $ & & Node & $ \bs{q} $ & Pipe. & $ \bs{c} $ \\
		\hline
		0   & $ 223 $  & L1-L13  & 12 & & 0           & $ 183 $   & $ \tilde{L}1-\tilde{L}11 $ & 12\\
		1-7  & $ -20 $ & C1-C3   & 12 & & 1-3,12,13   & $ -20 $   & $ \tilde{C}1-\tilde{C}4 $  & 12\\
		8-13 & $ -8 $  &         &    & & 4-6         & $ -8 $    &                 &\\
		&         &         &    & & 8-11        & $ -6.5 $  &                 &\\
		\hline\hline
		\multicolumn{9}{l}{The unit of $ \bs{q} $ is [$ m^3/s $]. The unit of $ \bs{c} $ is [$ \text{kg}^2m^{-8}s^{-2} $]}
	\end{tabular}
\end{table}

\begin{table}
	\renewcommand{\arraystretch}{1.3}
	\caption{Pipeline gas flows for Cases 1-3 of Network 1 under $ H_0 $ and $ H_1 $}
	\label{table_different_flow_net1}
	\centering
	\begin{tabular}{lcccccc}
		\hline\hline
		& \multicolumn{6}{c}{ Flows for each case [$ m^3/s $].}\\
		\hline
		& $ H_1 $ &  $ H_1 $ & $ H_1 $ & $ H_0 $ & $ H_0 $ & $ H_0 $\\
		{Pipe.} & Case 1  &  Case 2  &  Case 3 & Case 1  & Case 2 & Case 3\\
		\hline
		L1 	   & 223.0 & 223.0 & 223.0  & 223.0 & 223.0 & 223.0\\
		L2 	   & 100.0 & 100.0 & 100.0  & 100.0 & 100.0 & 100.0\\
		L3 	   & 75.0  & 75.0  & 75.0   & 75.0  & 75.0  & 75.0\\
		L4 	   & 50.0  & 50.0  & 50.0   & 50.0  & 50.0  & 50.0\\
		L5 	   & 25.0  & 25.0  & 25.0   & 25.0  & 25.0  & 25.0\\
		L6 	   & 98.0  & 98.0  & 98.0   & 98.0  & 98.0  & 98.0\\
		L7 	   & 73.0  & 73.0  & 73.0   & 73.0  & 73.0  & 73.0\\
		L8 	   & 48.0  & 48.0  & 48.0   & 48.0  & 48.0  & 48.0\\
		L9 	   & 20.4  & 20.0  & 21.4   & 21.4  & 19.2  & 20.4\\
		L10    & 6.2   & 12.0  & 6.7    & 6.7   & 11.2  & 6.2\\
		L11    & 4.4   & 0.0   & 0.9    & 0.9   & 3.2   & 4.4\\
		L12    & 8.0   & 4.0   & 3.5    & 3.5   & 8.0   & 8.0\\
		L13    & 19.5  & 20.0  & 18.5   & 18.5  & 20.8  & 19.5\\
		\hline
		\hline
	\end{tabular}
\end{table}

\begin{table}
	\renewcommand{\arraystretch}{1.3}
	\caption{Pipeline gas flows for Cases 1 and 2 of Network 2 under $ H_0 $ and $ H_1 $}
	\label{table_different_flow_net2}
	\centering
	\begin{tabular}{cccccccc}
		\hline\hline
		\multicolumn{8}{c}{ Flows for each case [$ m^3/s $]}\\
		\hline
		& $ H_1 $ &  $ H_1 $ & $ H_0 $ & & $ H_1 $ &  $ H_1 $ & $ H_0 $\\
		{Pipe.} & Case 1  &  Case 2 & Cases & {Pipe.} & Case 1  &  Case 2 & Cases\\
		\hline
		$ \tilde{L} $1 	   & 183.0 & 183.0 & 183.0 & $ \tilde{L} $7 	   & 6.5   & 6.5   & 7.9\\
		$ \tilde{L} $2 	   & 79.4  & 82.0  & 79.2  & $ \tilde{L} $8 	   & 19.5  & 19.5  & 20.9\\
		$ \tilde{L} $3 	   & 54.4  & 57.0  & 54.2  & $ \tilde{L} $9 	   & 28.5  & 26.0  & 28.7\\
		$ \tilde{L} $4 	   & 29.4  & 32.0  & 29.2  & $ \tilde{L} $10    & 53.5  & 51.0  & 53.7\\
		$ \tilde{L} $5 	   & 24.0  & 24.0  & 22.5  & $ \tilde{L} $11    & 78.5  & 76.0  & 78.7\\
		$ \tilde{L} $6 	   & 8.0   & 8.0   & 6.5\\
		\hline
		\hline
	\end{tabular}
\end{table}

\begin{table}
	\renewcommand{\arraystretch}{1.3}
	\caption{Directions of pipelines flows under $ H_0 $ for Cases 1-3 of Network 1}
	\label{table_def_tilde_A_net1}
	\centering
	\begin{tabular}{llllll}
	\hline\hline
	Pipeline & Start & End & Pipeline & Start & end\\
	Num. & Node & Node & Num. & Node & Node\\		
	\hline
	L1  & 0 & 1 & L9  & 8  & 9\\
	L2  & 1 & 2 & L10 & 8  & 10\\
	L3  & 2 & 5 & L11 & 9  & 11\\
	L4  & 5 & 6 & L12 & 11 & 12\\
	L5  & 6 & 7 & L13 & 12 & 13 \\
	L6  & 1 & 2 & C1  & 1  & 4\\
	L7  & 3 & 4 & C2  & 3  & 6\\
	L8  & 4 & 8 & C3  & 5  & 8\\
	\hline
	\hline
\end{tabular}
\end{table}

\begin{table}
	\renewcommand{\arraystretch}{1.3}
	\caption{Directions of pipelines flows under $ H_0 $ for Cases 1 and 2 of Network 2}
	\label{table_def_tilde_A_net2}
	\centering
	\begin{tabular}{llllll}
		\hline\hline
		Pipeline & Start & End & Pipeline & Start & end\\
		Num. & Node & Node & Num. & Node & Node\\		
		\hline
		L1  & 0 & 1 & L9  & 12 & 11\\
		L2  & 1 & 2 & L10 & 13 & 12\\
		L3  & 2 & 3 & L11 & 1  & 13\\
		L4  & 3 & 4 & C1  & 4  & 11\\
		L5  & 4 & 5 & C2  & 10 & 4 \\
		L6  & 6 & 7 & C3  & 5  & 6\\
		L7  & 9 & 8 & C4  & 7  & 8\\
		L8  & 11 & 10 \\
		\hline
		\hline
	\end{tabular}
\end{table}

In this section, we evaluate the performance of the Relaxed GLRT. 
To this end, we consider two natural gas networks in this paper. 
Fig. \ref{fig_model_gas_network_Bining} shows one 14-node natural gas network \cite{zhao2017coordinated}, called Network 1 while Fig. \ref{fig_model_gas_network_PESGM17} presents another 14-node natural gas network \cite{ojha2016solving}, called Network 2.
Pipelines marked $ C $ (in Network 1) or $ \tilde{C} $ (in Network 2) are changeable pipelines while the others are fixed. 
Table \ref{table_different_cases_net1} describes the states of changeable pipelines that the operator believes to be present ($ H_0 $) along with the actual pipeline settings ($ H_1 $) in three different cases (Cases 1-4) of Network 1.
Table \ref{table_different_cases_net2} describes the states of changeable pipelines in two cases (Cases 1 and 2) of Network 2.
Table \ref{table_variables} describes the gas injections/withdraws and the pipeline characteristic values used in our numerical investigations for both networks, which are fixed regardless of the changeable pipeline settings.
Table \ref{table_different_flow_net1} describes the actual gas flows for each pipeline under $ H_0 $ and $ H_1 $ for Network 1 while Table \ref{table_different_flow_net2} describes these flows for Network 2. 
Table \ref{table_def_tilde_A_net1} and Table \ref{table_def_tilde_A_net2} describe the directions of gas flows under $ H_0 $ for Network 1 and Network 2, respectively.
For Cases 1-4 of Network 1 under $ H_1 $, the direction of the flow of $ L4 $ is from 6 to 5, while the others are the same as in Table \ref{table_def_tilde_A_net1}. 
For Cases 1 and 2 of Network 2 under $ H_1 $, the direction of the flow of $ \tilde{L}7 $ is from 8 to 9, while the others are the same as in Table \ref{table_def_tilde_A_net2}. 
The relative size of the noise for each sensor is described using the relative standard derivation (RSD)
$ \tau = \sigma/\mu $
where $ \sigma $ and $ \mu $ are the standard derivation and the mean of the noisy measurement. RSD is often quoted as a percentage. We assume that each sensor has the same RSD. 
Our results assume pressure and injection/withdraw sensors at all nodes as well as flow sensors on all fixed pipelines to collect noisy measurements of pressures, injections/withdrawns, and flows. 
The metric we use to quantify the performance is the detection probability under $ 0.1\% $ false alarm probability for a given RSD.

If the detection probability is higher than $ 99.9\% $ in a specific case, we declare that the Relaxed GLRT obtains reliable performance for that case.
We employed Monte-Carlo simulations using $ 10^5 $ runs.

In Fig. \ref{fig_simu_perf_diff_RSD_0_10_q2v5}, we study $ P_{d} $ versus the RSD for Case 1 of Network 1 and Case 1 of Network 2. 
Fig. \ref{fig_simu_perf_diff_RSD_0_10_q2v5} shows that $ P_{d} $ increases as the RSD decreases.
For different $ {T_a} $, the RSD required to obtain good performance is also different. 
For $ {T_a}=35,40,45 $ and Case 1 of Network 1, the Relaxed GLRT provides good performance when RSD is lower than $ 6.6\%, 7\% $ and $ 7.5\% $, respectively.
Moreover, Fig. \ref{fig_simu_RSD10_q2_all} shows the $ P_{d} $ of all cases described in Tables \ref{table_different_cases_net1} and \ref{table_different_cases_net2}.
We observe that the Relaxed GLRT will always obtain good performance with a sufficiently large $ {T_a} $.
The minimum number of observations needed to obtain good performance for all cases is listed in the third row of Table \ref{table_fig_T_needed_list}. 

\begin{figure}
	\centering
	\subfloat[]{\includegraphics[width=0.35\textwidth]{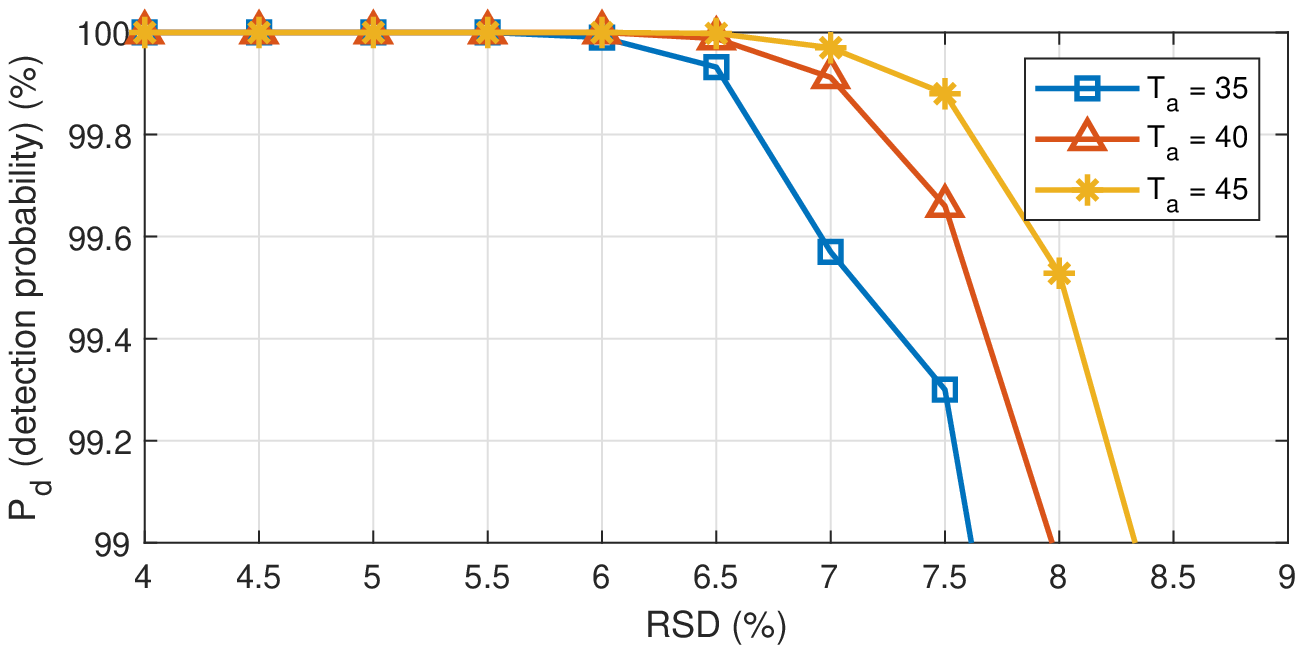}\label{fig_simu_perf_diff_RSD_0_10_q2v5_a}}
	\hfil
	\subfloat[]{\includegraphics[width=0.35\textwidth]{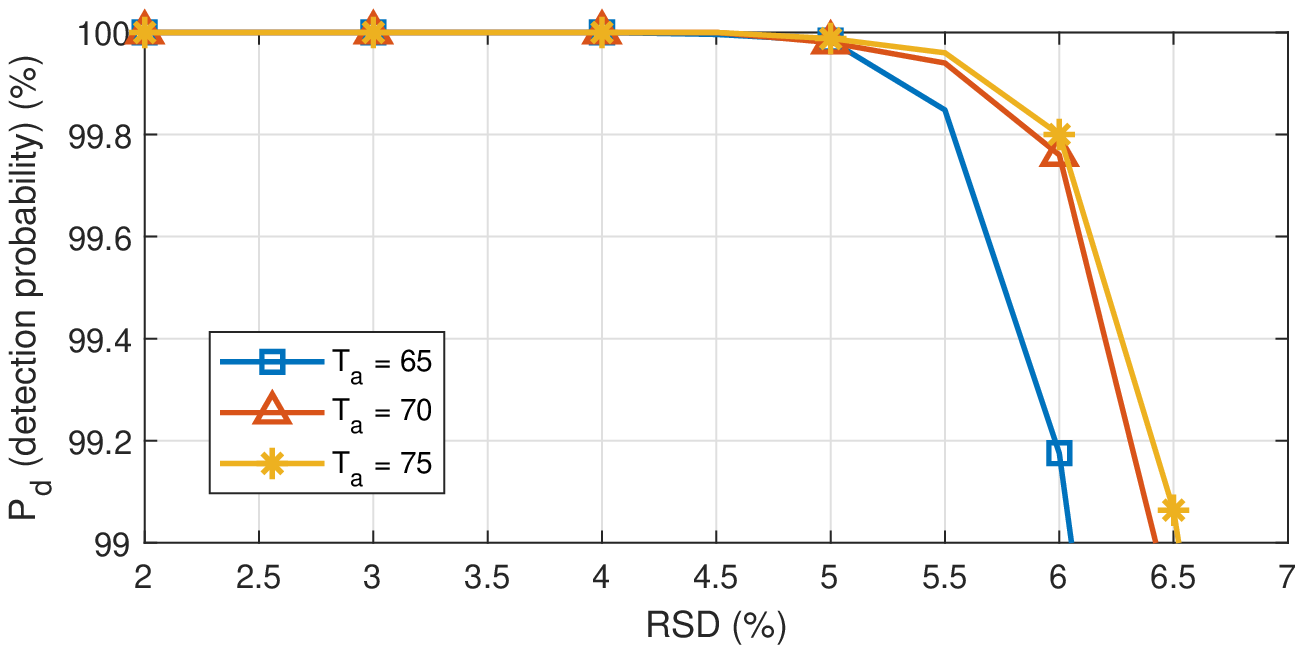}\label{fig_simu_perf_diff_RSD_0_10_q2v5_b}}
	\caption{Simulated $ P_{d} $ under different RSD values. (a) Case 1, Network 1; (b) Case 1, Network 2.}
	\label{fig_simu_perf_diff_RSD_0_10_q2v5}
\end{figure} 

\begin{figure}
	\centering
	\subfloat[]{\includegraphics[width=0.35\textwidth]{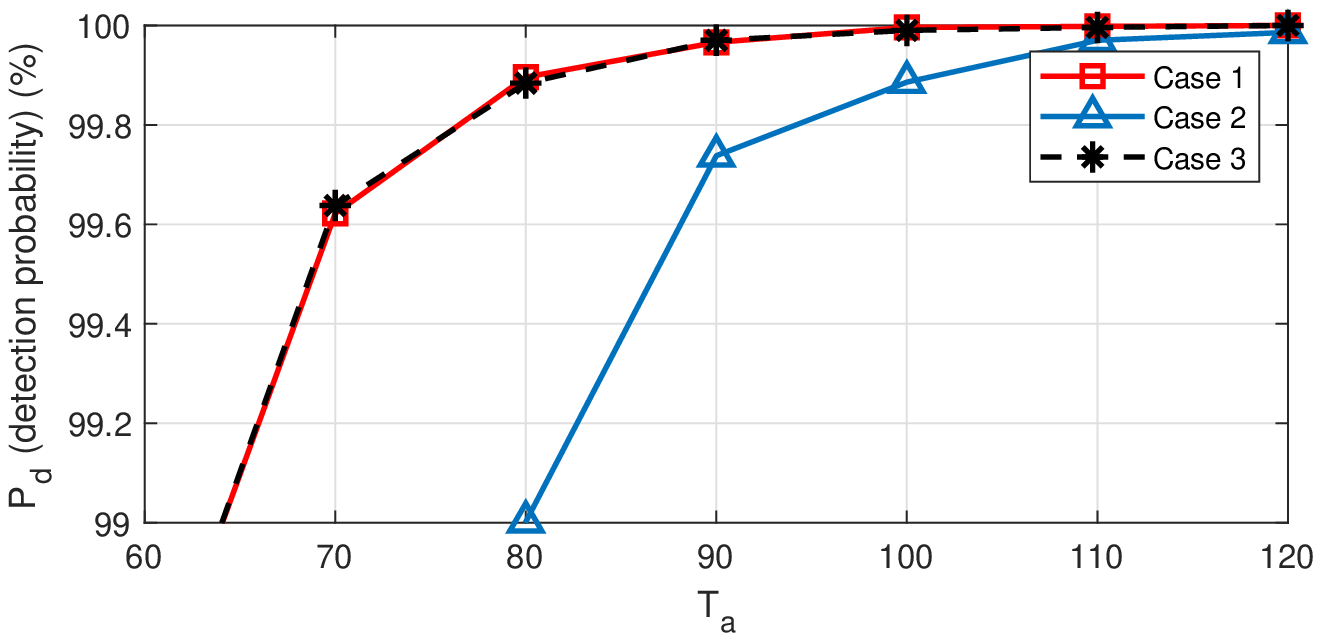}\label{fig_simu_RSD10_q2_all_a}}
	\hfil
	\subfloat[]{\includegraphics[width=0.35\textwidth]{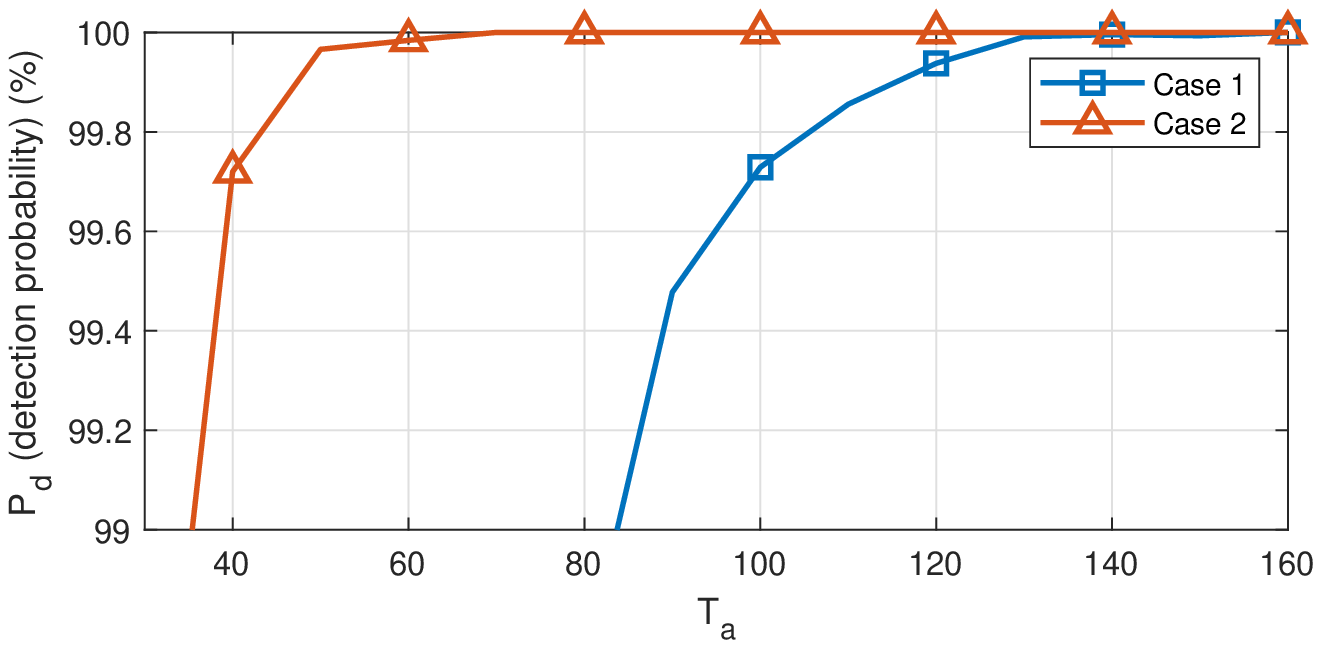}\label{fig_simu_RSD10_q2_all_b}}
	\caption{Simulated $ P_{d} $ of all cases. (a) Cases 1-3, Network 1 with $ 10\% $ RSD; (b) Cases 1 and 2, Network 2 with $ 8\% $ RSD.}
	\label{fig_simu_RSD10_q2_all}
\end{figure} 

\begin{table}
	\renewcommand{\arraystretch}{1.3}
	\caption{The detection probability $ P_{d} $ of simulations and the asymptotic performance for all cases}
	\label{table_fig_T_needed_list}
	\centering
	\begin{tabular}{lcccccc}
		\hline\hline
		& \multicolumn{3}{c}{Network 1} & & \multicolumn{2}{c}{Network 2}\\
		\cline{2-4} \cline{6-7}
		& Case 1 & Case 2 & Case 3 & & Case 1 & Case 2 \\	
		\hline	
		Min. $ {T_a} $ & $ 80 $ & $ 100 $ & $ 80 $ & & $ 114 $ & $ 47 $ \\
		$ P_{d} $ (Simu.)   & $ 99.90\% $ & $ 99.89\% $ & $ 99.90\% $ & & $ 99.89\% $ & $ 99.89\% $\\
		$ P_{d} $ (Asymp.)    & $ 99.65\% $ & $ 99.85\% $ & $ 99.65\% $ & & $ 99.87\% $ & $ 99.94\% $\\
		$ P_{d} $ Diff.      & $ 0.25\% $  & $ 0.04\% $  & $ 0.25\% $  & & $ 0.02\%  $ & $ 0.05\% $\\
		\hline
		\hline
		\multicolumn{7}{l}{$ 10\% $ RSD for Network 1; $ 8\% $ RSD for Network 2.}
	\end{tabular}
\end{table}

Next, we evaluated the minimum $ {T_a} $ to achieve good performance for the three cases under Network~1 and for the two cases under Network~2 and we listed these values of ${T_a}$ in Table \ref{table_fig_T_needed_list}. 
Table \ref{table_fig_T_needed_list} also lists both the asymptotic performance defined in Theorem \ref{Asymptotic_Performance_of_Standard_GLRT} for $ P_{d} $ and the simulated $ P_{d} $ for these cases.  These results show that the asymptotic performance is accurate for the value of $ {T_a} $ listed in Table \ref{table_fig_T_needed_list} which imply good performance. In particular, we see that for the given value of $ {T_a} $, the difference between the asymptotic performance and the simulated $ P_{d} $ is less than $ 0.25\% $ for all evaluated cases. Additionally, if we use an even larger value of $ T_a $, as shown in the yellow zones in Figures \ref{fig_asym_simu_perf} and \ref{fig_asym_simu_perf_net2}, the asymptotic performance is closer to the simulated performance.

\begin{figure}
	\centering
	\subfloat[]{\includegraphics[width=0.35\textwidth]{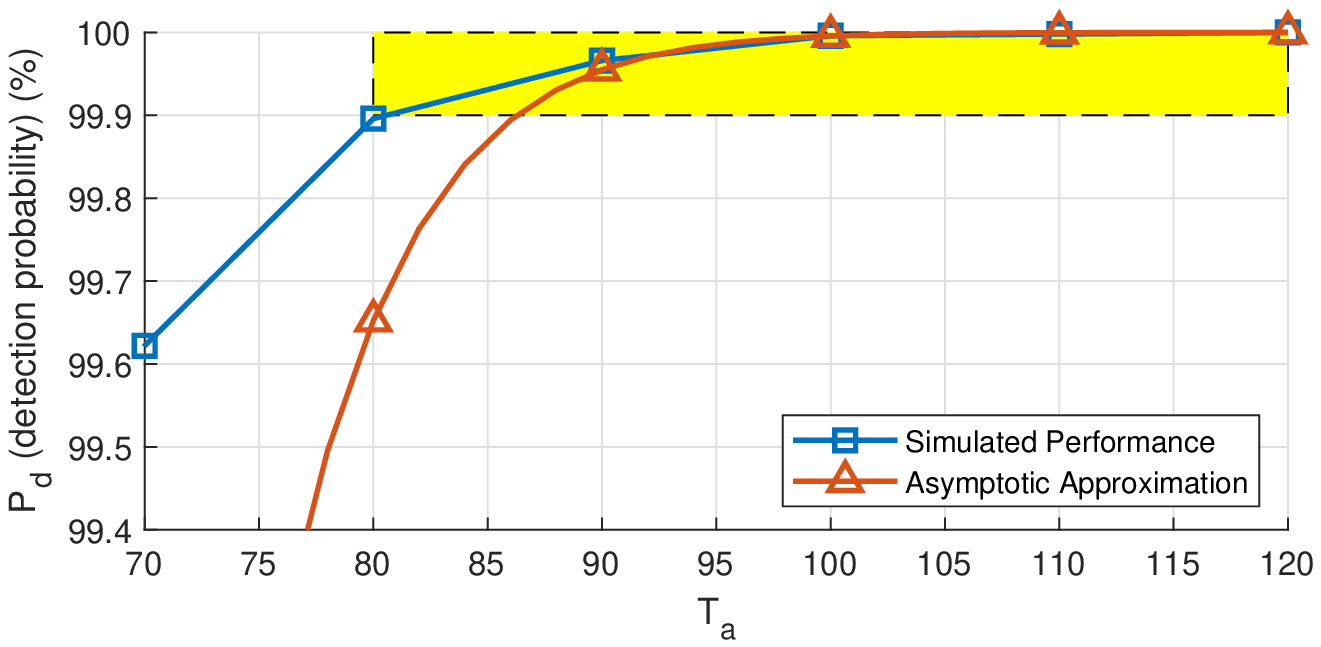}\label{fig_asym_simu_perf_a}}
	\hfil
	\subfloat[]{\includegraphics[width=0.35\textwidth]{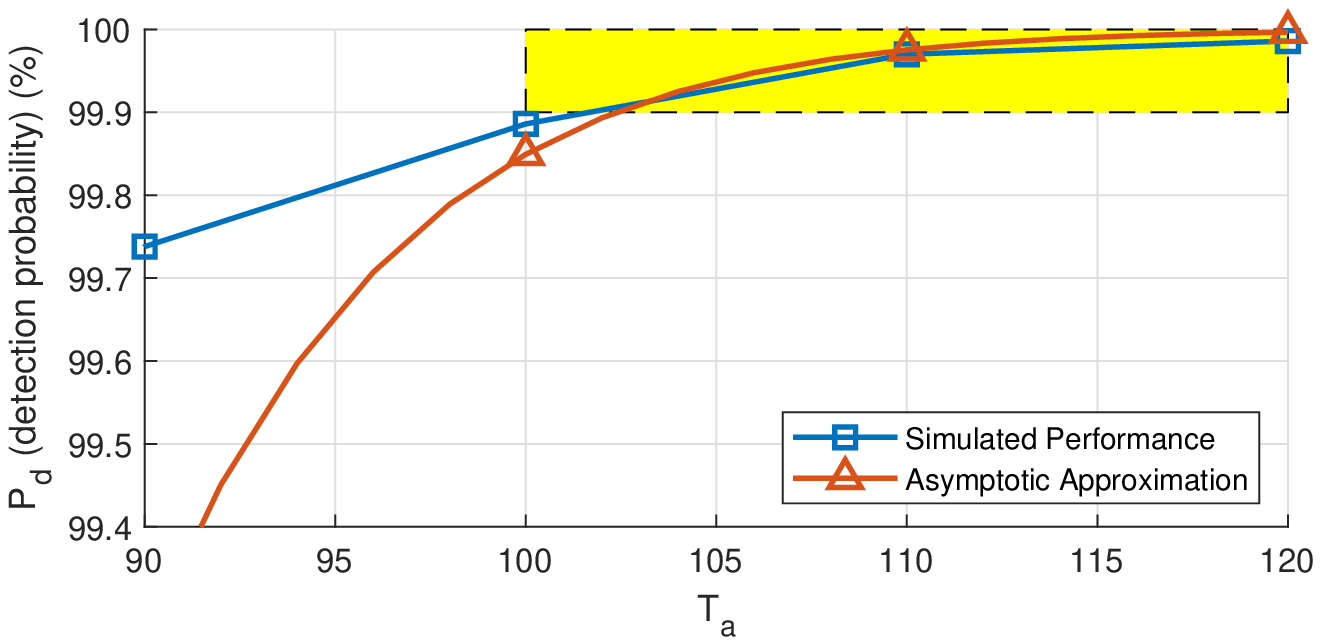}\label{fig_asym_simu_perf_b}}
	\hfil
	\subfloat[]{\includegraphics[width=0.35\textwidth]{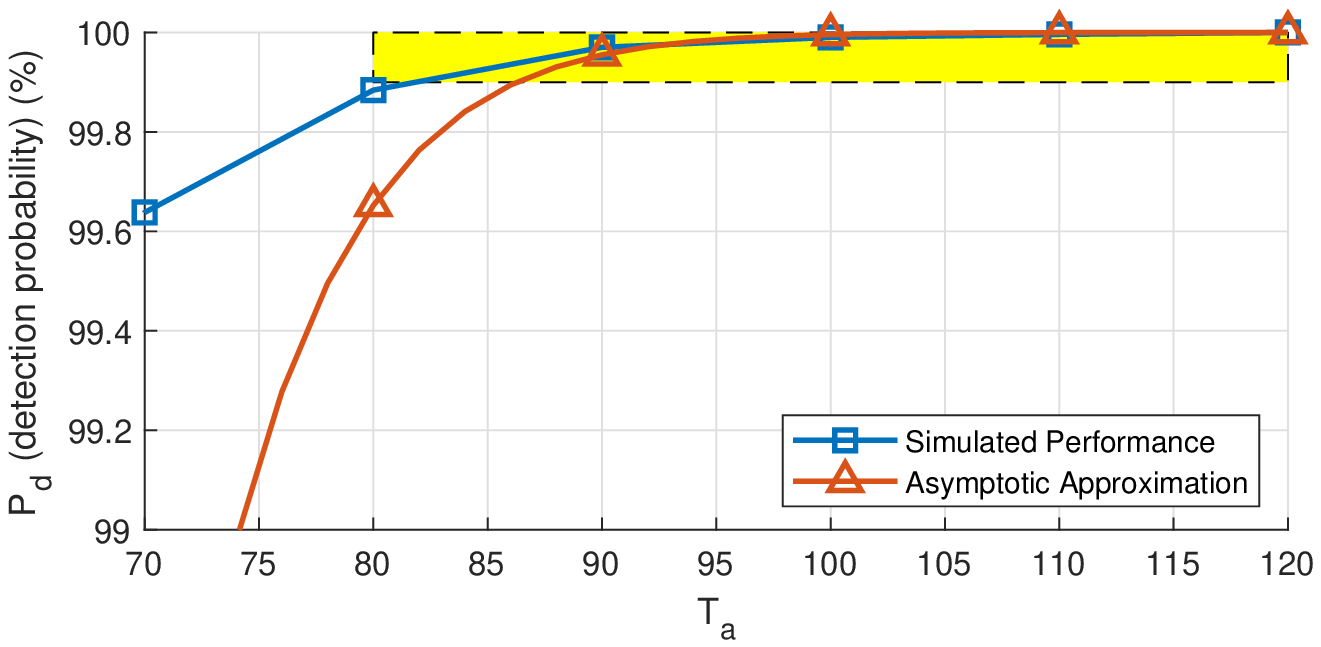}\label{fig_asym_simu_perf_c}}
	\caption{Asymptotic performances of the Relaxed GLRT compares to the simulated $ P_{d} $. (a) Case 1, Network 1; (b) Case 2, Network 1; (c) Case 3, Network 1.}
	\label{fig_asym_simu_perf}
\end{figure}

\begin{figure}
	\centering
	\subfloat[]{\includegraphics[width=0.35\textwidth]{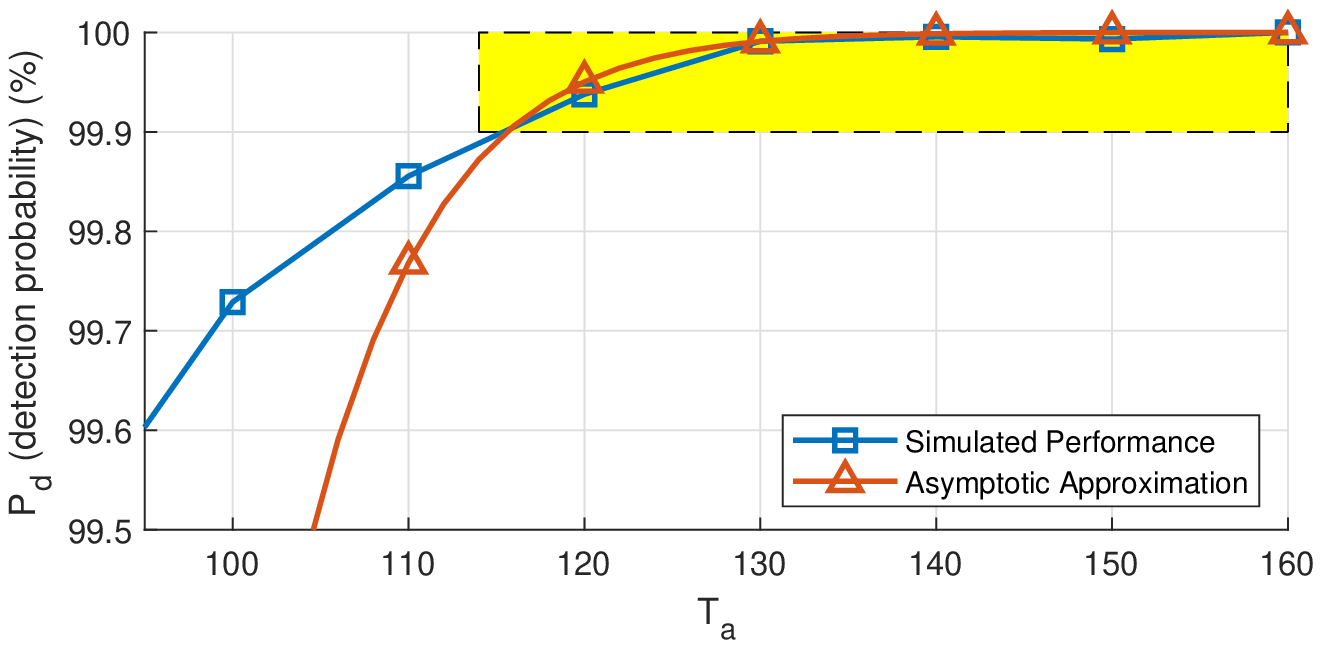}\label{fig_asym_simu_perf_net2_a}}
	\hfil
	\subfloat[]{\includegraphics[width=0.35\textwidth]{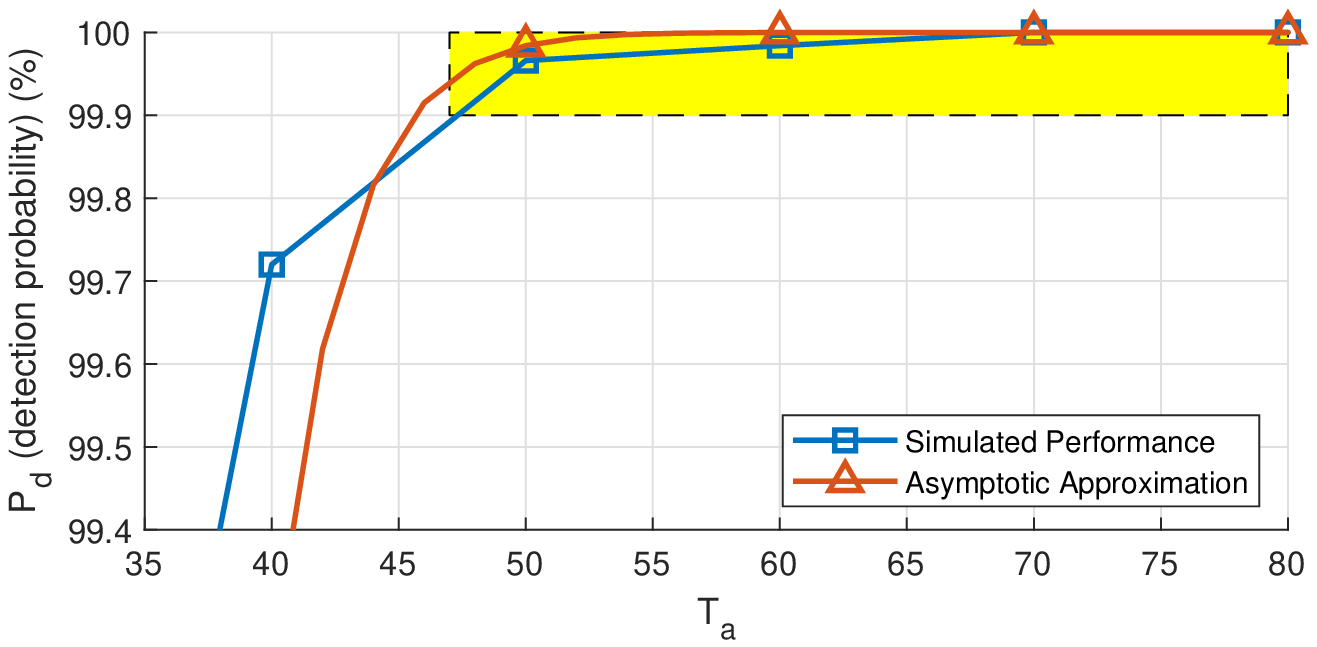}\label{fig_asym_simu_perf_net2_b}}
	\caption{Asymptotic performances of the Relaxed GLRT compares to the simulated $ P_{d} $. (a) Case 1, Network 2; (b) Case 2, Network 2.}
	\label{fig_asym_simu_perf_net2}
\end{figure} 
We give numerical result to verify that the solution to each optimization in (33)  (the relaxed problem) is unit rank based on (24) and (26) for a sufficiently large $ T_a $. For this numerical investigation, when an eigenvalue is smaller than $ 10^{-4} $, we regard it as zero. Fig. \ref{fig:rankxversust} shows the rank of the optimal $ \bs{X} $ (number of non-zero eigenvalues) for Network 1 in Case 1 and 2 with $ 5\% $ RSD. Fig. \ref{fig:rankxversust} shows that the rank of the optimal $ \bs{X} $ is always $ 1 $ for $ T_a > 100 $.

\begin{figure}
	\centering
	\includegraphics[width=0.7\linewidth]{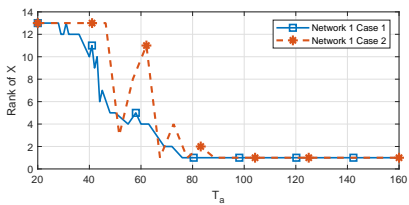}
	\caption{The rank of the optimal $ \bs{X} $ for Network 1 Case 1 and 2.}
	\label{fig:rankxversust}
\end{figure}

\begin{figure}
	\centering
	\subfloat[]{\includegraphics[width=0.35\textwidth]{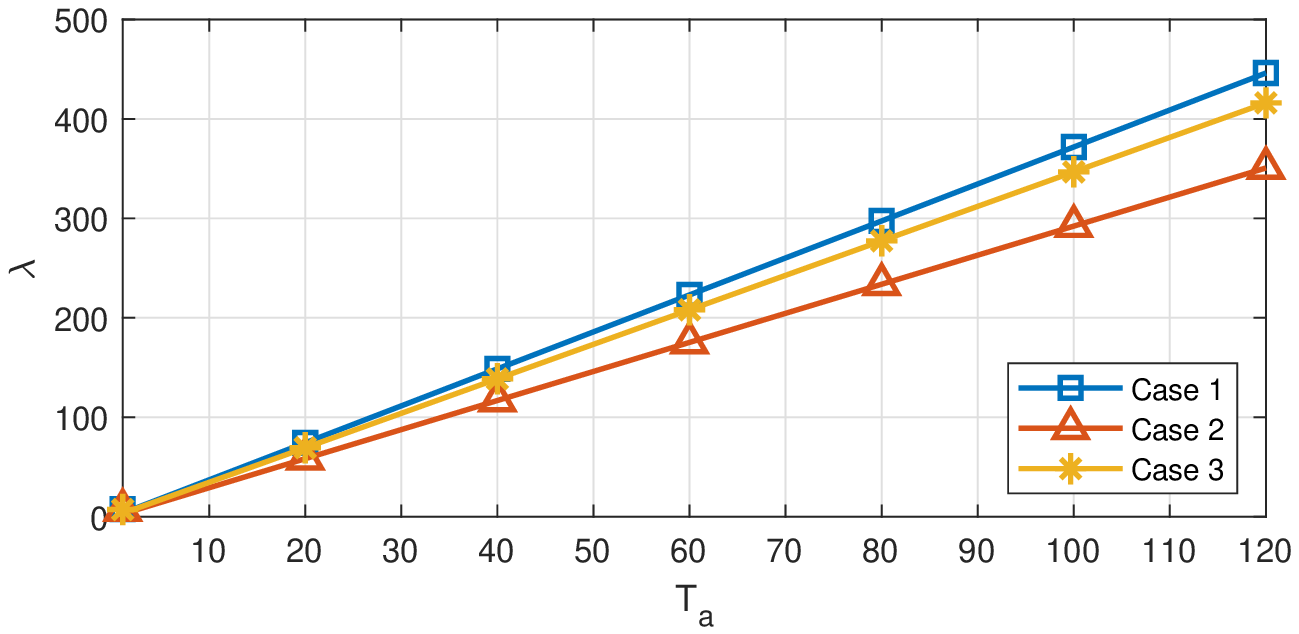}\label{fig_lambda_evaluation_a}}
	\hfil
	\subfloat[]{\includegraphics[width=0.35\textwidth]{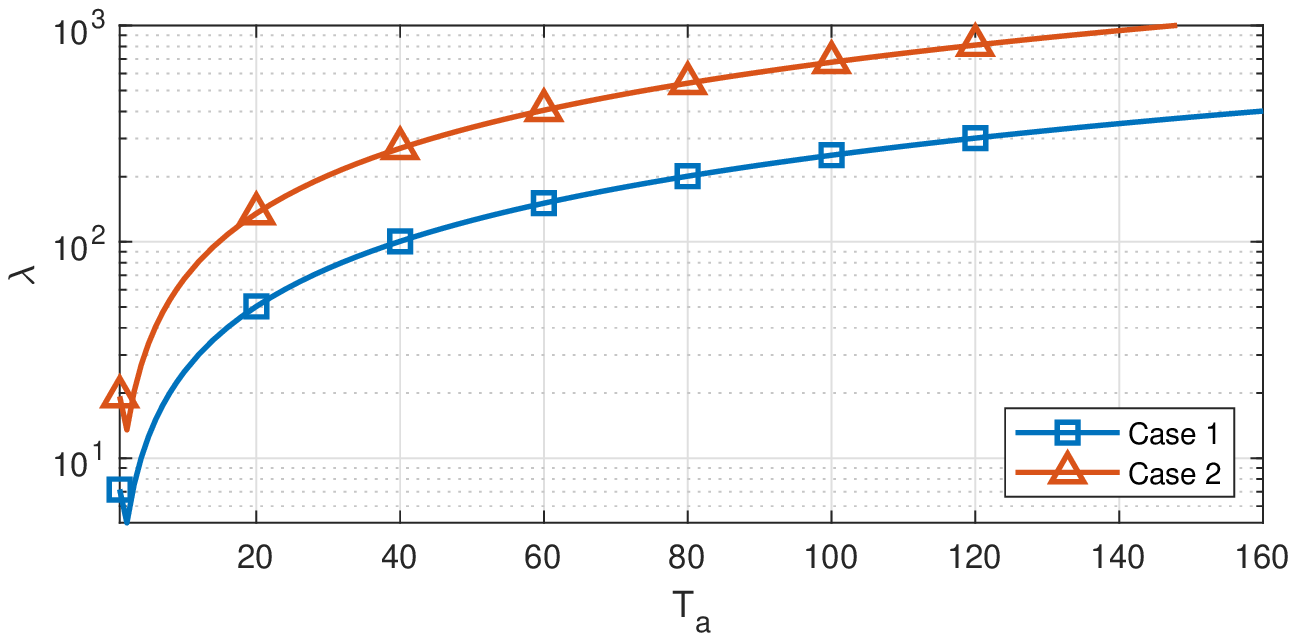}\label{fig_lambda_evaluation_b}}
	\caption{Evaluation of $ \lambda $. (a) Cases 1-3, Network 1, $ 10\% $ RSD; (b) Cases 1 and 2, Network 2, $ 8\% $ RSD.}
	\label{fig_lambda_evaluation}
\end{figure}

The detection probability $ P_{d} $ is very sensitive to the exact pipeline settings under $ H_1 $ and $ H_0 $. 
In practice we must choose $ {T_a} $ sufficiently large such that the most difficult topology change will be detected properly.
In order to choose $ T_a $, we first find the worst case by employing the asymptotic performance in Theorem \ref{Asymptotic_Performance_of_Standard_GLRT}, and then use it to find the required $ {T_a} $. 
For example, based on (\ref{def_lambda}), Fig. \ref{fig_lambda_evaluation} plots the values of $ \lambda $ for the cases of Networks 1 and 2 described in Tables \ref{table_different_cases_net1}-\ref{table_def_tilde_A_net2}.
These evaluations show that Case 2 of Network 1 and Case 1 of Network 2 are the worst of the considered cases since they have the smallest $ \lambda $. 
Based on (\ref{Pd_standard}) and (\ref{def_lambda}), the minimum value of $ {T_a} $ to obtain good performance for the hardest case is determined by
\begin{align}
\overline{{T_a}} = &\lfloor Q^{-1}_{S^2}\left( 0.999, [Q^{-1}_{S^2}(0,10^{-3})]^{0.5} \right)/\lambda\rfloor, \label{best_of_lambda}
\end{align}
with $ \lambda $ chosen from the hardest case and defined in \eqref{def_lambda}. 
In (\ref{best_of_lambda}), $ \lfloor x\rfloor $ denotes the floor operator, and $ Q^{-1}_{S^2}(P_d, \rho) $ denotes the inverse Marcum Q-function with $ S^2 $ degrees of freedom defined in \cite{gil2014asymptotic}. 
Using (\ref{best_of_lambda}), we know that the values of $ {T_a} $ to obtain good performance for the worst case of Network 1 and Network 2 are $ 100 $ and $ 114 $ which agree well with simulation results. 

Let $ P_{d|\bs{A}^*\ne\bs{A}_{H_1}} $ denote the probability of detection given $ \bs{A}^*\ne\bs{A}_{H_1} $.
The total detection probability is equal to
\begin{align}
P_{d} = P_{d|\bs{A}^*\ne\bs{A}_{H_1}}\text{Pr}(\bs{A}^*\ne\bs{A}_{H_1}) + P_{d|\bs{A}^*=\bs{A}_{H_1}}\text{Pr}(\bs{A}^*=\bs{A}_{H_1}). \label{total_pdr}
\end{align}
Our asymptotic performance ignores the first term in the RHS of (\ref{total_pdr}) since $ \text{Pr}(\bs{A}^*=\bs{A}_{H_1}) $ goes to 1 for large $ T $. Table~\ref{table_prob_choose_right_A} presents the simulated $ \text{Pr}(\bs{A}^*=\bs{A}_{H_1}) $ as well as the two terms in the RHS of (\ref{total_pdr}), separately. Table~\ref{table_prob_choose_right_A} shows that $ P_{d|\bs{A}^*\ne\bs{A}_{H_1}}\text{Pr}(\bs{A}^*\ne\bs{A}_{H_1}) $ goes to zero as $ T $ increases. 

\begin{table*}
	\renewcommand{\arraystretch}{1.3}
	\caption{The probability of $ \bs{A}=\bs{A}_{H_1} $ and corresponding $ P_{d} $.}
	\label{table_prob_choose_right_A}
	\centering
	\begin{tabular}{c|lccccccc}
		\hline\hline
		Pipeline Setting& $ {T_a} $ & $ 50 $ & $ 60 $ & $ 70 $ & $ 80 $ & $ 90 $ & $ 100 $\\
		\hline
		Network 1 & $ P_{d|\bs{A}^*\ne\bs{A}_{H_1}}\text{Pr}(\bs{A}^*\ne\bs{A}_{H_1}) $ & $ 0.0069\% $ & $ 0.0013\% $ & $ 0.0002\% $ & $ <0.0001\% $ & $ <0.0001\% $ & $ <0.0001\% $\\
		Case 1& $ P_{d|\bs{A}^*=\bs{A}_{H_1}}\text{Pr}(\bs{A}^*=\bs{A}_{H_1}) $ & $ 98.11\% $ & $ 99.12\% $ & $ 99.62\% $ & $ 99.74\% $ & $ 99.84\% $ & $ 99.99\% $\\
		RSD$ =10\% $& $ \text{Pr}(\bs{A}^*=\bs{A}_{H_1}) $ & $ 99.16\% $ & $ 99.64\% $& $ 99.86\% $& $ 99.92\% $& $ 99.94\% $& $ 99.99\% $\\
		\hline
		Network 1 & $ P_{d|\bs{A}^*\ne\bs{A}_{H_1}}\text{Pr}(\bs{A}^*\ne\bs{A}_{H_1}) $ & $ 0.42\% $ & $ 0.06\% $ & $ 0.03\% $ & $ 0.01\% $ & $ 0.01\% $ & $ 0.01\% $\\
		Case 2 & $ P_{d|\bs{A}^*=\bs{A}_{H_1}}\text{Pr}(\bs{A}^*=\bs{A}_{H_1}) $ & $ 98.49\% $ & $ 99.38\% $ & $ 99.58\% $ & $ 99.76\% $ & $ 99.86\% $ & $ 99.98\% $\\
		RSD$ =10\% $& $ \text{Pr}(\bs{A}^*=\bs{A}_{H_1}) $ & $ 99.34\% $ & $ 99.74\% $& $ 99.82\% $& $ 99.90\% $& $ 99.92\% $& $ 99.99\% $\\
		\hline
		\hline
	\end{tabular}
\end{table*}

Table \ref{table:algorithmRunTime} shows the run times of the algorithm described in Section \ref{eq:effiTopoVerAlgo} along with the run times needed for solving the Relaxed GLRT by enumeration for all cases of Network 1 and 2 with $ T = 100 $ and $ 10\% $ RSD.
The threshold $ \epsilon $ is set to $ Q_{L+N}^{-1}(0,P_{fa}) $ in which  $ Q_{L+N}^{-1}(0,P_{fa}) $ was defined after \eqref{eq:defMarcumQ} based on \cite{monticelli1983reliable}.
These results show that for all cases the efficient algorithm is much faster than solving the Relaxed GLRT by enumeration. 
We see that when $ H_1 $ is true, the algorithm stops quickly for all cases as we stated in Section \ref{eq:effiTopoVerAlgo}. 

To evaluate the performance of our placement approach, 
For each Monte-Carlo run, the cost of placing any sensor is independently and randomly set from $ 1 $ to $ 70 $. 
Let $ D $ and $ D^* $ denote the cost obtained by our placement approach and the cost of the optimal placement obtained by exhausting search, respectively.
Let $ D/D^* $ denote the approximation ratio. 
Table \ref{table:placeResult} shows the approximation ratio $ D/D^* $ and the run times for 4 different gas networks. Table \ref{table:placeResult} shows that our placement approach is very close to the optimal solution for all cases. Table \ref{table:placeResult} also shows that our placement approach is very efficient since, on average we check \eqref{eq:app_pre_condition} less than $ 21 $ times while the exhausting search needs to check \eqref{eq:app_pre_condition} $ 2.20\times 10^6 $ times.

\begin{table}
	\renewcommand{\arraystretch}{1.3}
	\caption{The average run times of the algorithm}
	\label{table:algorithmRunTime}
	\centering
	\begin{tabular}{lccccccc}
		\hline\hline
		& \multicolumn{4}{c}{Network 1} & & \multicolumn{2}{c}{Network 2}\\
		\cline{2-5} \cline{7-8}
		& Case 1 & Case 2 & Case 3 & Case 4 & & Case 1 & Case 2 \\	
		\hline	
		$ H_0 $ true & $ 2.56s $ & $ 2.89s $ & $ 2.01s $ & $ 3.44s $ & & $ 2.80s $ & $ 1.57s $ \\
		$ H_1 $ true & $ 0.54s $ & $ 0.62s $ & $ 1.48s $ & $ 0.64s $ & & $ 1.43s $ & $ 0.62s $\\
		Enum.   & $ 5.30s $ & $ 5.56s $ & $ 4.14s $ & $ 5.36s $ & & $ 9.76s $ & $ 9.84s $\\
		\hline
		\hline
	\end{tabular}
\end{table}

\begin{table}
	\renewcommand{\arraystretch}{1.3}
	\caption{The performance of our placement approach for different networks.}
	\label{table:placeResult}
	\centering
	\begin{tabular}{ccccc}
		\hline\hline
		Networks & \includegraphics[width=0.45in]{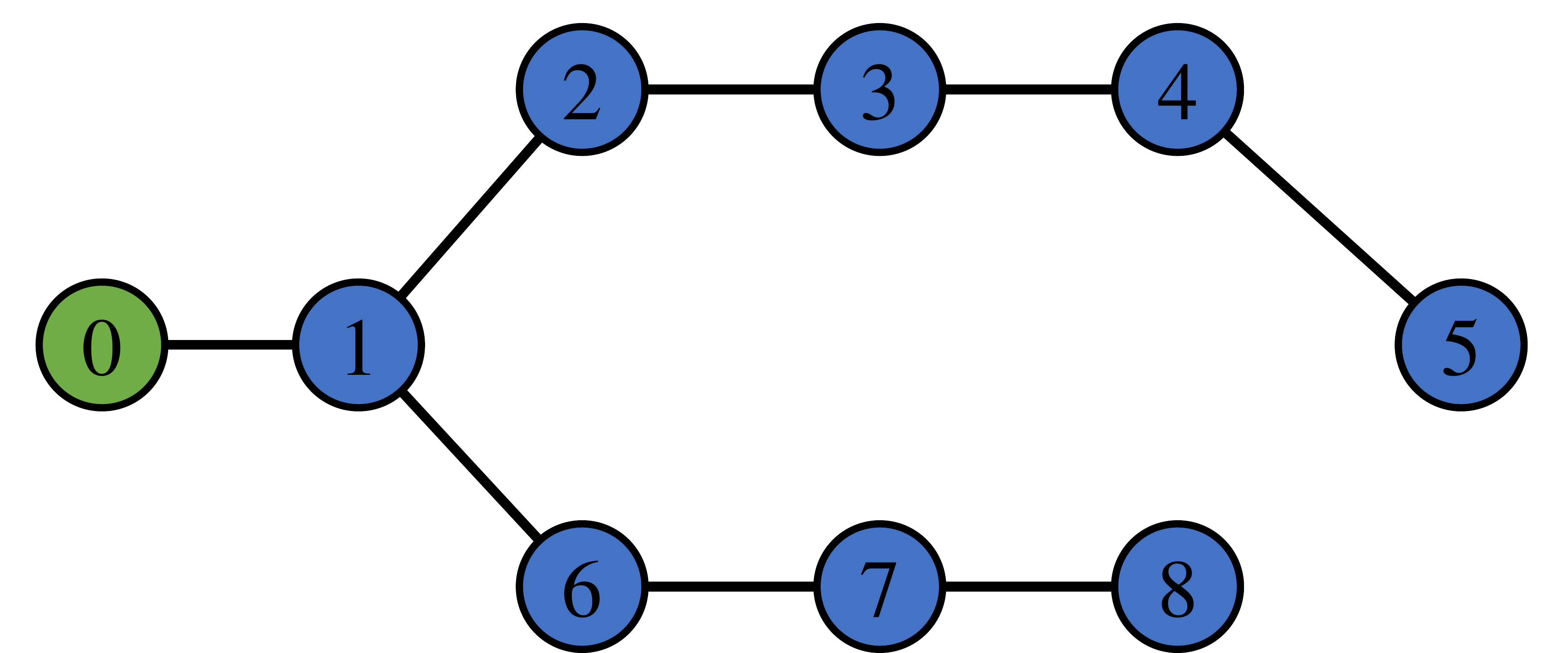} & \includegraphics[width=0.45in]{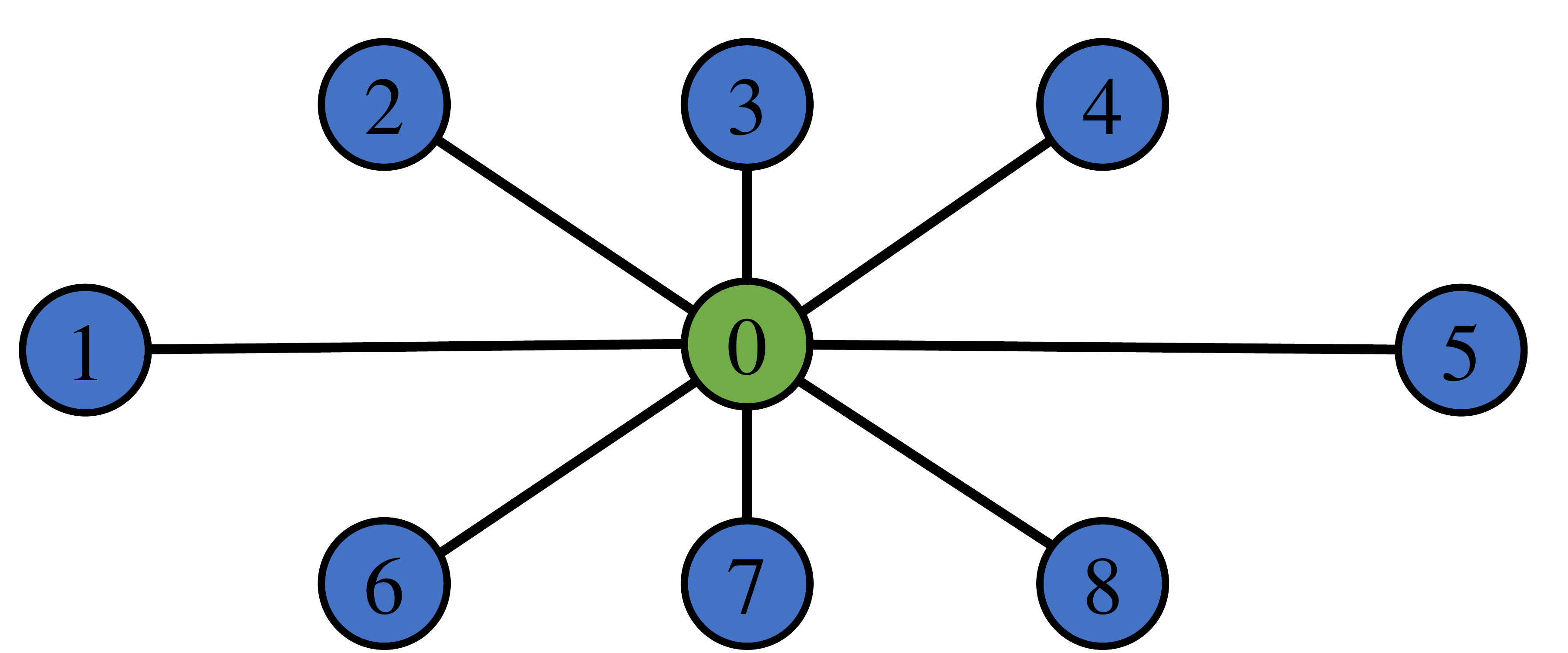} & \includegraphics[width=0.45in]{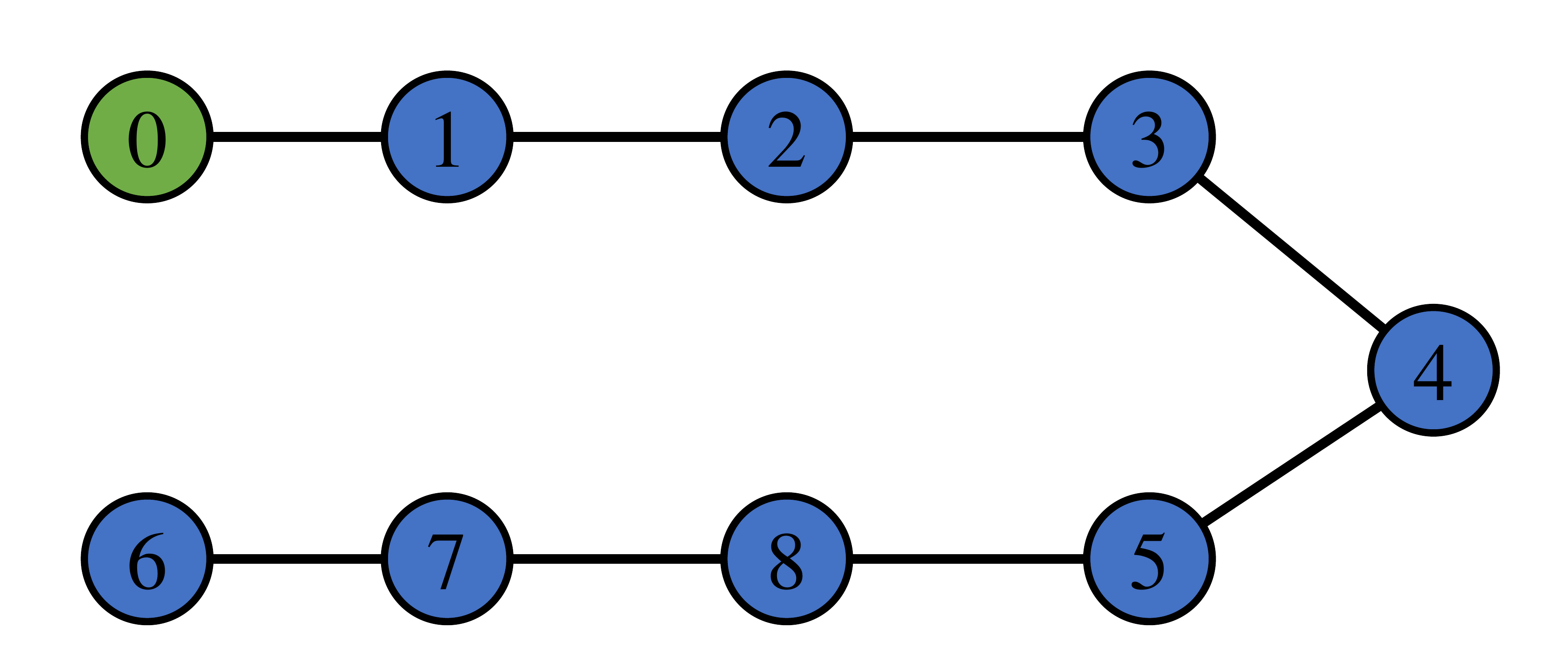} & \includegraphics[width=0.45in]{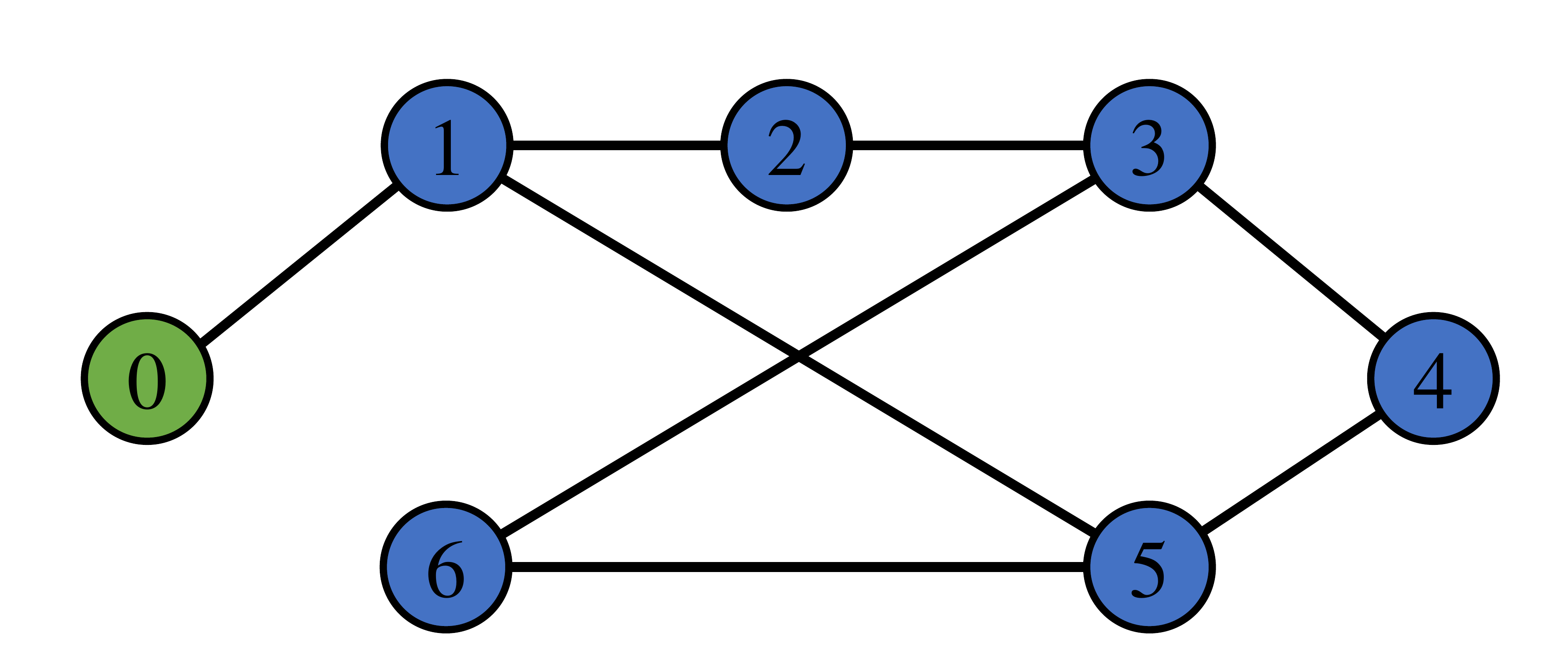} \\
		\hline	
		$ D/D^* $ & $ 1.1398 $ & $ 1.2263 $ & $ 1.1224 $ & $ 1.2703 $\\
		Run time & $ 0.1438 $s & $ 0.1945 $s & $ 0.1586 $s & $ 0.1655 $s \\
		\hline
		\hline
		\multicolumn{5}{l}{Exhausting search takes $ 2123.70 $s .}
	\end{tabular}
\end{table}

The numerical results illustrate the following properties of the Relaxed GLRT:
	1) {\bf The Relaxed GLRT provides good performance for sufficiently large $ {T_a} $ for all cases studied}: Simulation results show that $ P_{d} $ increases as $ {T_a} $ is increased in each case. 
	2) {\bf The performance of the Relaxed GLRT depends on the pipeline settings under $ H_0 $ and $ H_1 $}:
	Table \ref{table_fig_T_needed_list} shows that $ P_{d} $ varies for different cases and networks given fixed $ {T_a} $.
	3) {\bf The asymptotic performance proposed in Theorem \ref{Asymptotic_Performance_of_Standard_GLRT} is accurate for the values of $ {T_a} $ needed to achieve good performance.}
	4) {\bf In cases where we require good performance, we use $ \lambda $ to find out the most difficult pipeline change to detect, then the required number of observations to achieve good performance is given by (\ref{best_of_lambda}).}

\section{conclusions and discussion}

In this paper, we describe a new algorithm for employing sensor measurements to perform topology verification of natural gas networks.  Our algorithm is based on relaxing some optimization problems in the classical statistical hypothesis testing approach called the GLRT. The GLRT generally works well when the optimization problems involved in computing it are convex, but it is difficult to guarantee good performance when some of the optimization problems are nonconvex. However, in our application the GLRT will involve some nonconvex optimizations that correspond to estimating some continuous valued parameters describing the mathematical models of some pressures and flows in the natural gas network. Our new algorithm removed some constraints to produce convex optimizations for the estimation of these continuous valued variables, while ensuring identical performance to the GLRT for an appropriate number of sensor observations. Hence, our new algorithm is much efficient and reliable since solving it only involves convex optimization programming. We derive new closed-form expressions for the asymptotic performance of our algorithm in terms of the probability of detection. We have not seen any derivations in the literature of the asymptotic performance for the cases we consider here in this paper, which involve estimations of discrete variables which describe the natural gas network topology. The asymptotic results also describe the required number of observations needed to achieve reliable performance. We also provide sufficient conditions for an appropriate sensor placement. 
The approach employed to derive the closed-form asymptotic performance expressions is pretty general and can be used to solve any GLRT which involves both continuous and discrete quantities that must be estimated or any modified GLRT which uses relaxed versions of the optimizations employing these estimations. 

As this is the first paper on this topic of topology verification of natural gas networks, we focus on the most basic setting.  For example, we do not consider cases where the sensors may fail and give faulty measurements or where the communications which send the sensor data may fail.  We also do not consider cases where the sensor measurements might be manipulated by attackers. 
We intend to study these cases in future work.  
The baseline provided in this paper is very helpful as a start for such work, since it describes very favorable performance with no attacks or failures. 
If we can find acceptable complexity topology identification approaches which can identify attacks and failures and use this information to appropriately process the compromised and uncompromised data to achieve performance close to the unattacked performance, see \cite{zhang2018approaches}, then these approaches will be very useful for practical deployment.  

\appendices
\section{The Proof of Theorem \ref{thm_relaxed_is_exact}: Asymptotic equivalence of the relaxed ML to the original ML} \label{app:proofOfTheoremAsympEquivalent}

Define a function $ {h}:\mathbb{R}^{S\times S}\rightarrow\mathbb{R}^S $ by $ {h}(\bs{X}) = \begin{bmatrix}\bs{X}_{1,1}&\bs{X}_{2,2}&\dots&\bs{X}_{S,S}	\end{bmatrix}^T $.
In the relaxed ML \eqref{relaxed_problem}, the constraints $ \tilde{f}(m,{\bs{A}},\bs{X}) = \Tr{\bs{Z}(m,{\bs{A}})\bs{X}}+\bs{b}_mp_0^2 $ confine only diagonal elements of $ \bs{X} $, because $  \bs{Z}(m,{\bs{A}}) $ is a diagonal matrix which was defined as
\begin{align}
\bs{Z}(m,{\bs{A}}) &= \diag{\begin{bmatrix} \bs{B}_{[m,:]}&-\bs{c}_m\bs{e}_{m,{1}\times L}&0 \end{bmatrix}}.
\end{align} 
where $ m=1,\dots,L $, $ \bs{B}_{[m,:]} $ denotes the $ m^{th} $ row of $ \bs{B} $, $ \bs{c}_m $ denotes $ m^{th} $ entry of the pipeline characteristic vector $ \bs{c} $, and $ \bs{e}_{m,1\times L} $ is an $ 1\times L $ row vector with all zero entries except for the $ m^{th} $ which is 1. 
Then, the diagonal elements of $ \bs{X} $ live in the solution space of the following linear equations
\begin{align}
\begin{bmatrix}
\tilde{f}(1,\bs{A},\bs{X})\\
\vdots\\
\tilde{f}(L,\bs{A},\bs{X})
\end{bmatrix} = \bs{0}.
\end{align}
Let $ \Pi(\bs{A}) = \{ h(\bs{X}) | \tilde{f}(m,\bs{A},\bs{X}) = 0, m=1,\dots,L \} $.
Then, the relaxed ML \eqref{relaxed_problem} becomes
\begin{align}
\xi(\bs{A}) =  \sup_{\begin{subarray}{c} {h}(\bs{X})\in\Pi(\bs{A}),\hpt \bs{X}\ge 0,\hpt \bs{X} = \bs{X}^T \\ \end{subarray}}  C-\Tr{\bs{M}(\bs{A},\bs{v})\bs{X}} \label{relaxed_prob_equiv}.
\end{align}
Define $ \bs{1} $ as an all one column vector whose the dimension is equal to the number of observations. Define the vectors $ \overline{\bs{\delta}}_p = \bs{1}\otimes \bs{\delta}_p $, $ \overline{\bs{\delta}}_q = \bs{1}\otimes \bs{\delta}_q $, $ \overline{\bs{\delta}}_\phi = \bs{1}\otimes \begin{bmatrix}
\bs{\delta}_F^T&\bs{\delta}_C^T
\end{bmatrix}^T $, $ \overline{\bs{I}}_L = \bs{1}\otimes \bs{I}_L $, $ \overline{\bs{I}}_N = \bs{1}\otimes \bs{I}_N $, and $ \overline{\bs{A}} = \bs{1}\otimes \bs{A}^T $ where $ \otimes $ denotes the kronecker product.
Define the full observation vectors as $ \overline{\bs{p}} = (\tilde{\bs{p}}^T[1],\hpt\cdots,\tilde{\bs{p}}^T[T_a])^T $, $ \overline{\bs{q}} = (\tilde{\bs{q}}^T[1],\hpt\cdots,\tilde{\bs{q}}^T[T_a])^T $, and $ \overline{\bs{\phi}} = (\tilde{\bs{\phi}}^T[1],\hpt\cdots,\tilde{\bs{\phi}}^T[T_a])^T $.
Define the full noise vectors as $ \overline{\bs{n}}_p = (\bs{n}_p^T[1],\hpt\cdots,\hpt\bs{n}_p^T[T_a])^T $, $ \overline{\bs{n}}_q = (\bs{n}_q^T[1],\hpt\cdots,\hpt\bs{n}_q^T[T_a])^T $, and $ \overline{\bs{n}}_\phi = (\bs{n}_\phi^T[1],\hpt\cdots,\hpt\bs{n}_\phi^T[T_a])^T $.
Then, the original ML \eqref{original_problem} becomes
\begin{align}
&\mu(\bs{A})=\sup_{  \bs{\omega}\odot\bs{\omega}\in\Pi(\bs{A})}  C - \frac{1}{2}\left[\left\|\frac{1}{\sigma_\phi}
\overline{\bs{\delta}}_\phi\odot\left(
\overline{\bs{\phi}} - \overline{\bs{I}}_L\bs{\phi}\right)\right\|^2 \right.\notag\\ 
&\left.+  \left\|\frac{1}{\sigma_q}\overline{\bs{\delta}}_q\odot\left(\overline{\bs{q}} - \overline{\bs{A}}
\bs{\phi}\right)\right\|^2
+ \left\|\frac{1}{\sigma_p}\overline{\bs{\delta}}_p\odot\left(\overline{\bs{p}} - \overline{\bs{I}}_N\bs{p}\right)\right\|^2\right]. \label{original_prob_equiv}
\end{align} 
Proposition 5 in \cite{so2010probabilistic} shows that $ \mu(\bs{A}) = \xi(\bs{A}) $ when the following conditions are satisfied simultaneously 
\begin{align}
\eta_{\min}\left[ \overline{\bs{I}}^T\diag{\overline{\bs{\delta}}_p}\overline{\bs{I}} \right] &> \left\| \overline{\bs{I}}^T\diag{\overline{\bs{\delta}}_p}^T\overline{\bs{n}}_p \right\|_{\infty},\label{exact_condition_1}\\ 
\eta_{\min}\left[ \overline{\bs{A}}^T\diag{\overline{\bs{\delta}}_q}\overline{\bs{A}} \right] &> \left\| \overline{\bs{A}}^T\diag{\overline{\bs{\delta}}_q}^T\overline{\bs{n}}_q \right\|_{\infty}, \label{exact_condition_2}\\ 
\eta_{\min}\left[ \overline{\bs{I}}^T\diag{\overline{\bs{\delta}}_\phi}\overline{\bs{I}} \right] &> \left\| \overline{\bs{I}}^T\diag{\overline{\bs{\delta}}_\phi}^T\overline{\bs{n}}_\phi \right\|_{\infty}. \label{exact_condition_3}
\end{align}
Using $ \eta_{\min}({T_a}\bs{A}) = {T_a}\eta_{\min}(\bs{A}) $, \eqref{exact_condition_1}-\eqref{exact_condition_3} become
\begin{align}
{T_a} &> \left\| \diag{\bs{\delta}_p}\sum_{t=1}^{T_a}\bs{n}_p[t] \right\|_{\infty},\\ 
{T_a}\eta_{\min}\left[ \bs{A}\hpt\diag{\bs{\delta}_q}\hpt\bs{A}^T \right] &> \left\| \bs{A}\hpt\diag{\bs{\delta}_q}\sum_{t=1}^{T_a}\bs{n}_q[t] \right\|_{\infty},\\ 
{T_a} &> \left\| \diag{\bs{\delta}_\phi}\sum_{t=1}^{T_a}\bs{n}_\phi[t] \right\|_{\infty}, 
\end{align}
or, taken together, we have \eqref{exact_relaxed_final_condi} in Theorem \ref{thm_relaxed_is_exact}.

\ifCLASSOPTIONcaptionsoff
\newpage
\fi

\bibliographystyle{IEEEtran}

\begin{thebibliography}{10}
\providecommand{\url}[1]{#1}
\csname url@samestyle\endcsname
\providecommand{\newblock}{\relax}
\providecommand{\bibinfo}[2]{#2}
\providecommand{\BIBentrySTDinterwordspacing}{\spaceskip=0pt\relax}
\providecommand{\BIBentryALTinterwordstretchfactor}{4}
\providecommand{\BIBentryALTinterwordspacing}{\spaceskip=\fontdimen2\font plus
\BIBentryALTinterwordstretchfactor\fontdimen3\font minus
  \fontdimen4\font\relax}
\providecommand{\BIBforeignlanguage}[2]{{%
\expandafter\ifx\csname l@#1\endcsname\relax
\typeout{** WARNING: IEEEtran.bst: No hyphenation pattern has been}%
\typeout{** loaded for the language `#1'. Using the pattern for}%
\typeout{** the default language instead.}%
\else
\language=\csname l@#1\endcsname
\fi
#2}}
\providecommand{\BIBdecl}{\relax}
\BIBdecl

\bibitem{weimer2012distributed}
J.~Weimer, S.~Kar, and K.~H. Johansson, ``Distributed detection and isolation
  of topology attacks in power networks,'' in \emph{Proceedings of the 1st
  international conference on High Confidence Networked Systems}.\hskip 1em
  plus 0.5em minus 0.4em\relax ACM, 2012, pp. 65--72.

\bibitem{liang2017framework}
G.~Liang, S.~R. Weller, J.~Zhao, F.~Luo, and Z.~Y. Dong, ``A framework for
  cyber-topology attacks: Line-switching and new attack scenarios,'' \emph{IEEE
  Transactions on Smart Grid}, 2017.

\bibitem{sevlian2018outage}
R.~A. Sevlian, Y.~Zhao, R.~Rajagopal, A.~Goldsmith, and H.~V. Poor, ``Outage
  detection using load and line flow measurements in power distribution
  systems,'' \emph{IEEE Transactions on Power Systems}, vol.~33, no.~2, pp.
  2053--2069, 2018.

\bibitem{cavraro2015distribution}
G.~Cavraro, R.~Arghandeh, G.~Barchi, and A.~von Meier, ``Distribution network
  topology detection with time-series measurements,'' in \emph{2015 IEEE Power
  \& Energy Society Innovative Smart Grid Technologies Conference
  (ISGT)}.\hskip 1em plus 0.5em minus 0.4em\relax IEEE, 2015, pp. 1--5.

\bibitem{deka2016estimating}
D.~Deka, S.~Backhaus, and M.~Chertkov, ``Estimating distribution grid
  topologies: A graphical learning based approach,'' in \emph{Power Systems
  Computation Conference (PSCC), 2016}.\hskip 1em plus 0.5em minus 0.4em\relax
  IEEE, 2016, pp. 1--7.

\bibitem{gao2013method}
Y.~Gao, Z.~Zhang, W.~Wu, and H.~Liang, ``A method for the topology
  identification of distribution system,'' in \emph{2013 IEEE Power \& Energy
  Society General Meeting}.\hskip 1em plus 0.5em minus 0.4em\relax IEEE, 2013,
  pp. 1--5.

\bibitem{tian2016mixed}
Z.~Tian, W.~Wu, and B.~Zhang, ``A mixed integer quadratic programming model for
  topology identification in distribution network,'' \emph{IEEE Transactions on
  Power Systems}, vol.~31, no.~1, pp. 823--824, 2016.

\bibitem{deka2018structure}
D.~Deka, S.~Backhaus, and M.~Chertkov, ``Structure learning in power
  distribution networks,'' \emph{IEEE Transactions on Control of Network
  Systems}, vol.~5, no.~3, pp. 1061--1074, 2018.

\bibitem{lehmann2006testing}
E.~L. Lehmann and J.~P. Romano, \emph{Testing statistical hypotheses}.\hskip
  1em plus 0.5em minus 0.4em\relax Springer Science \& Business Media, 2006.

\bibitem{wang2019topology}
Z.~Wang and R.~S. Blum, ``Topology attack detection in natural gas delivery
  networks,'' in \emph{2019 53rd Annual Conference on Information Sciences and
  Systems (CISS)}.\hskip 1em plus 0.5em minus 0.4em\relax IEEE, 2019, pp. 1--6.

\bibitem{osiadacz1987simulation}
A.~Osiadacz, ``Simulation and analysis of gas networks,'' 1987.

\bibitem{poor2013introduction}
H.~V. Poor, \emph{An introduction to signal detection and estimation}.\hskip
  1em plus 0.5em minus 0.4em\relax Springer Science \& Business Media, 2013.

\bibitem{zeitouni1992generalized}
O.~Zeitouni, J.~Ziv, and N.~Merhav, ``When is the generalized likelihood ratio
  test optimal?'' \emph{IEEE Transactions on Information Theory}, vol.~38,
  no.~5, pp. 1597--1602, 1992.

\bibitem{luo2010semidefinite}
Z.-Q. Luo, W.-K. Ma, A.~M.-C. So, Y.~Ye, and S.~Zhang, ``Semidefinite
  relaxation of quadratic optimization problems,'' \emph{IEEE Signal Processing
  Magazine}, vol.~27, no.~3, pp. 20--34, 2010.

\bibitem{so2010probabilistic}
A.~M.-C. So, ``Probabilistic analysis of the semidefinite relaxation detector
  in digital communications,'' in \emph{Proceedings of the twenty-first annual
  ACM-SIAM symposium on Discrete algorithms}.\hskip 1em plus 0.5em minus
  0.4em\relax SIAM, 2010, pp. 698--711.

\bibitem{hu2000gradient}
J.~Hu and R.~S. Blum, ``A gradient guided search algorithm for multiuser
  detection,'' \emph{IEEE Communications Letters}, vol.~4, no.~11, pp.
  340--342, 2000.

\bibitem{moore2010constrained}
T.~J. Moore and B.~M. Sadler, ``Constrained hypothesis testing and the
  cram{\'e}r-rao bound,'' in \emph{Sensor Array and Multichannel Signal
  Processing Workshop (SAM), 2010 IEEE}.\hskip 1em plus 0.5em minus 0.4em\relax
  IEEE, 2010, pp. 113--116.

\bibitem{kay1998fundamentals}
S.~M. Kay, ``Fundamentals of statistical signal processing, volume ii:
  Detection theory,'' 1998.

\bibitem{sun2010monotonicity}
Y.~Sun, {\'A}.~Baricz, and S.~Zhou, ``On the monotonicity, log-concavity, and
  tight bounds of the generalized marcum and nuttall $ q $-functions,''
  \emph{IEEE Transactions on Information Theory}, vol.~56, no.~3, pp.
  1166--1186, 2010.

\bibitem{gil2014asymptotic}
A.~Gil, J.~Segura, and N.~M. Temme, ``The asymptotic and numerical inversion of
  the marcum q-function,'' \emph{Studies in Applied Mathematics}, vol. 133,
  no.~2, pp. 257--278, 2014.

\bibitem{moore2008maximum}
T.~J. Moore, B.~M. Sadler, and R.~J. Kozick, ``Maximum-likelihood estimation,
  the cram{\'e}r--rao bound, and the method of scoring with parameter
  constraints,'' \emph{IEEE Transactions on Signal Processing}, vol.~56, no.~3,
  pp. 895--908, 2008.

\bibitem{zhao2017coordinated}
B.~Zhao, A.~J. Conejo, and R.~Sioshansi, ``Coordinated expansion planning of
  natural gas and electric power systems,'' \emph{IEEE Transactions on Power
  Systems}, 2017.

\bibitem{ojha2016solving}
A.~Ojha, V.~Kekatos, and R.~Baldick, ``Solving the natural gas flow problem
  using semidefinite program relaxation,'' in \emph{Proc. IEEE Power \& Energy
  Society General Meeting, Chicago, IL}, 2016.

\bibitem{zhang2018approaches}
J.~Zhang, R.~S. Blum, and H.~V. Poor, ``Approaches to secure inference in the
  internet of things: Performance bounds, algorithms, and effective attacks on
  iot sensor networks,'' \emph{IEEE Signal Processing Magazine}, vol.~35,
  no.~5, pp. 50--63, 2018.

\end{thebibliography}

\end{document}